\documentclass[a4paper]{amsart}

\usepackage[T1]{fontenc}
\usepackage{amsfonts,amsmath,amssymb,amsxtra,amscd,amsthm}
\usepackage{mathrsfs}
\usepackage[Symbol]{upgreek}
\usepackage[nice]{nicefrac}
\usepackage{cancel}
\usepackage{ae}
\usepackage{dsfont}
\usepackage[all]{xy}
\usepackage{graphicx}
\usepackage{layout,color}
\usepackage[multiple]{footmisc}

\def\plus#1#2{\vrule height#1pt width0pt depth#2pt}
\def\@{\hskip.8pt}
\def\?{\hskip.3pt}
\newcommand{\abs}[1] {|#1|}
\newcommand{\internal}{\rule{0.5em}{0.1mm}\rule{0.1mm}{0.7em}\ }
\newcommand{\rarw}[1]{\overset{#1}{\longrightarrow}}

\renewcommand{\geq}{\geqslant}

\def\a{\alpha}
\def\A{\mathcal{A}}
\def\b{\beta}
\def\B{\mathcal{B}}

\def\C{\mathcal{C}}
\def\Chi{\text{\lower-2pt\hbox{$\chi$}}}
\def\eps{\varepsilon}
\def\E{\mathcal E}
\def\F{\mathscr{F}}
\def\g{\gamma}
\def\G{\Gamma}
\def\H{\mathcal{H}}
\def\Ham{\mathscr{H}}
\def\I{\mathcal{I}}
\def\k{\kappa}
\def\l{\lambda}
\def\L{\mathcal{L}}
\def\Lagr{\mathscr{L}}
\def\LL{\h{\L\;}\hskip-.3em}
\def\HH{\h{\H\;}\hskip-.3em}

\def\narc#1{\text{\@ \scriptsize $(#1)$}}
\def\Ns{\text{\@ \tiny $(s)$}}
\def\NS{\text{\@ \tiny $(s+1)$}}

\def\ns{\narc{s}}
\def\nS{\narc{s+1}}

\def\P_#1{\mathcal P_{#1}\/}

\def\r{\varrho}
\def\R{\mathds{R}}
\def\s{\sigma}
\def\SS{\mathcal{S}}

\def\th{\vartheta}
\def\Th{\Theta}
\def\thL{\th_{\hskip-1.2pt \Lagr}}

\def\u{\upsilon}
\def\U{\Upsilon}
\def\vphi{\varphi}
\def\vs{\varsigma}
\def\V{\mathcal{V}_{n+1}}
\def\Vg{V\?(\g)}
\def\VV{\mathfrak V}
\def\X{\mathfrak X}
\def\W{\mathfrak W}
\def\w{\omega}

\def\z{\zeta}
\def\Z{\mathcal Z\/}

\def\h#1{\hat{#1}}
\def\j#1{j_1\/(#1)}
\def\p{\wp\/(\g)}
\def\Ag{A\@(\h\g)}

\def\Hg{\H\@(\h\g)}
\def\Vg{V\?(\g)}
\def\Wg{V\?(\h\g)}

\def\gxi{\g^{\phantom{|}}_{\?\xi}}
\def\hgxi{\h\g^{\phantom{|}}_{\?\xi}}
\def\atilde{\raise-.8ex\hbox to0pt{\tiny$\sim\hss$} \a{}}

\def\d#1/d#2{\frac{d\/#1}{d\/#2}}
\def\de#1/de#2{\frac{\partial\/#1}{\partial\/#2}}
\def\SD#1/de#2/de#3{\ifx#2 \frac{\plus02\partial^{\@\@2}#1}
    {\plus90\partial\@#3^{\@2}} \else\frac{\plus02\partial^{\@\@2}#1}
    {\partial\?#2\partial\?#3}\fi}

\def\D#1/D#2{\frac{D\/#1}{D\/#2}}

\def\De#1/de#2{\textstyle{\text{\Large$\de{#1}/de{#2}$}}}
\def\sD#1/de#2/de#3{\textstyle{\text{\Large$\SD{#1}/de{#2}/de{#3}$}}}
\def\sd#1/de#2/de#3{\ifx#2 \frac{\plus02\partial^{\@\@2}#1}{\plus70\partial\@#3^{\@2}} \else\frac{\plus02\partial^{\@\@2}#1}{\partial\?#2\partial\?#3}\fi}

\def\DE{\tilde\partial\?}

\def\rank{\operatorname{rank}}
\def\const{\text{const.}}

\theoremstyle{plain}
\newtheorem {remark} {Remark} [section]
\newtheorem {theorem}{Theorem}[section]
\newtheorem {proposition}{Proposition}[section]
\newtheorem {definition} {Definition} [section]
\newtheorem {corollary}{Corollary}[section]

\catcode`\@=11
\@addtoreset{equation}{section}
\def\theequation{\thesection.\arabic{equation}}
\catcode`\@=12

\catcode`\@=11
\renewenvironment{subequations}{%
  \refstepcounter{equation}%
  \protected@edef\theparentequation{\theequation}%
  \setcounter{parentequation}{\value{equation}}%
  \setcounter{equation}{0}%
  \def\theequation{\theparentequation\hspace{1pt}\alph{equation}}%
  \ignorespaces
}{%
  \setcounter{equation}{\value{parentequation}}%
  \ignorespacesafterend
} \catcode`\@=12

\def\ALPH{\renewcommand{\labelenumi}{\alph{enumi})}}
\def\nn{\nonumber}

\def\Tondo{\par\vskip1pt\noindent$\bullet\;$ }
\def\Ref#1#2{\if#2)\ref{#1}#2\else\ref{#1}\@#2\fi}

\def\wg#1{\@\delta q^{#1}{}_{|\g}}
\def\wgs#1{\@\delta q^{#1}{\!\@}_{|\g^{(s)}}}

\def\BIG#1_#2{\Big#1_{\!\@\lower-1pt\hbox{$\scriptstyle#2$}}}
\def\BIGG#1_#2{\bigg#1_{\!\lower-1.4pt\hbox{$\scriptstyle#2$}}}
\def\BIGGR#1_#2{\biggr#1_{\!\lower-1.4pt\hbox{$\scriptstyle#2$}}}
\def\RIGHT#1_#2{\right#1_{\!\lower-1.4pt\hbox{$\scriptstyle#2$}}}

\def\Eps#1{\eps_{\lower1pt\hbox{$\scriptstyle #1$}}}
\def\e#1#2#3#4{\ifx#1_e_{(#2)}^{\;\,#4}\else e^{(#2)}_{\;#4}\fi}
\def\arc#1#2{\big(\@#1,[#2]\@\big)}

\usepackage{fullpage}
\begin{document}

\title[Geometric constrained variational calculus. I.]{Geometric Constrained Variational Calculus.\\ I. - Piecewise smooth extremals.}

\author{E. Massa}
\address{Dipartimento di Matematica -- Universit\`a di Genova\\
Via Dodecaneso, $35$ -- $16146$ Genova (Italia)}
\email{massa\@@\@dima.unige.it}

\author{D. Bruno}
\address{Department of Advanced Robotics -- Istituto Italiano di Tecnologia \\
Via Morego, $30$ - $16163$ Genova (Italia)}
\email{danilo.bruno\@@\@iit.it}

\author{G.Luria}
\address{DIPTEM Sez. Metodi e Modelli Matematici -- Universit\`a di Genova\\
         Piazzale Kennedy, Pad. D -- $16129$ Genova (Italia)}
\email{luria\@@\@dime.unige.it}

\author{E. Pagani}
\address{Dipartimento di Matematica -- Universit\`a di Trento \\
       Via Sommarive, $14$ - $38050$ Povo di Trento (Italia)}
\email{pagani\@@\@science.unitn.it}

\begin{abstract}
A geometric setup for constrained variational calculus is presented. The analysis deals with the study of the extremals of an action functional defined on
piecewise differentiable curves, subject to differentiable, non--holonomic constraints. Special attention is paid to the tensorial aspects of the theory. As
far as the kinematical foundations are concerned, a fully covariant scheme is developed through the introduction of the concept of {\em infinitesimal
control\/}. The standard classification of the extremals into \emph{normal\/} and \emph{abnormal\/} ones is discussed, pointing out the existence of an
algebraic algorithm assigning to each admissible curve a corresponding \emph{abnormality index\/}, related to the co--rank of a suitable linear map. Attention
is then shifted to the study of the first variation of the action functional. The analysis includes a revisitation of Pontryagin's equations and of the
Lagrange multipliers method, as well as a reformulation of Pontryagin's algorithm in hamiltonian terms. The analysis is completed by a general result,
concerning the existence of finite deformations with fixed endpoints.
\end{abstract}

\subjclass[2010]{49J,70F25,37J}
\keywords{Constrained calculus of variations, minimality.}

\maketitle

\section{Introduction}
Calculus of variations has a very old origin, dating back to the pioneering contributions of Euler, Lagrange and Weierstrass\@\?%
\footnote%
{Many good texts provide a modern exposition of the classical theory. Among them, we point out those by Bolza \cite{Bolza}, Tonelli \cite{Tonelli},
Caratheodory \cite{Caratheodory}, Bliss \cite{Bliss}, Morse \cite{Morse}, Lanczos \cite{Lanczos}, Gelfand-Fomin \cite{Gelfand}, Rund \cite{Rund},
Giaquinta-Hildebrandt \cite{Giaquinta} and Sagan \cite{Sagan}.}.

The 20th century witnessed the development of the classical theory towards two different main directions: on the one hand, the curves on which a given
functional is defined, and is required to be minimized, were imposed to satisfy a set of differential (or non--holonomic) constraints. This topic took the name
of \emph{constrained calculus of variations\/}. On the other hand, mathematicians accounted for the possibility, more directly related to \emph{optimization
theory\/}, of considering also unilateral constraints, described by \emph{inequalities\/}.

Both aspects were subsequently merged into \emph{control theory\/}. Here, an important role is played by the celebrated Pontryagin Maximum Principle, which
provides the necessary conditions that must be satisfied by a minimizing solution of the given functional. In this connection, besides the original
contributions of Pontryagin \cite{Pontryagin}, more recent developments have been given by Bellman \cite{Bellman}, Hestenes \cite{Hestenes}, Cesari
\cite{Cesari}, Alekscseev, Tikhomirov and Fomin \cite{Alekscseev-Tikhomirov-Fomin}, Schattler and Ledzewicz \cite{Schattler-Ledzewicz}.

Recently, Griffiths \cite{Griffiths}, Hsu \cite{Hsu}, Sussmann and his school \cite{Sussmann}, Sussmann and Liu \cite{SL}, Montgomery \cite{Montgomery},
Agrachev and his school \cite{Agrachev}, Gracia
\cite{Gracia}, Munoz--Lecanda and his school \cite{Barbero-Linan1,Barbero-Linan2}, Jimenez, Yoshimura, Ibort, de Leon and Martin de Diego
\cite{Ibort-Pena-Salmoni, deLeon-Jimenez-deDiego, Jimenez-Yoshimura} have significantly contributed to develop differential geometric approaches to
constrained calculus of variations and to control theory.

\medskip
This paper deals with constrained variational calculus in the sense defined above.
We present a fresh geometric approach to the subject, extending known results and suggesting new ones. Attention is focussed on the construction of a fully
covariant scheme, suited to the study of arbitrary non--linear constraints and embodying the possibility of \emph{piecewise differentiable}
extremals. The geometric setup and terminology are borrowed from non--holonomic mechanics \mbox{\cite{MP2,MVB}}.

Throughout the discussion we shall deal with \emph{parameterized curves\/}, namely curves whose parameterization is fixed once for all, up to an additive
constant; geometrically, this means dealing with sections of a fiber bundle, called the \emph{event space\/}.
In the resulting context, the constraints are described by an embedded submanifold of the first jet--bundle of the event space,
intuitively representing the totality of {\em admissible velocities\/}.

The study of the extremals is performed using the standard techniques of variational calculus, namely discussing the zeroes of the first variation of the
action functional with respect to deformations with fixed endpoints. Old though the whole argument may seem, some questions are still unsettled: among others,
the geometrical nature of the concept of {\em normality\/} of a section; the relation between normality and deformability;
the creation of fully covariant scheme embodying the kinematical and dynamical aspects of the theory.

The arguments are organized as follows: in Sec. \!2 the geometric tools needed in the subsequent discussion are presented. After a few preliminary remarks, the
infinitesimal deformations of an admissible section $\@\g\@$ of the event space are discussed via a revisitation of the well known \emph{variational
equation\/}. The novelty consists in the introduction of a suitable transport law for vertical vector fields along $\@\g\@$, yielding a covariant
characterization of the ``true'' degrees of freedom involved in the description of the most general admissible infinitesimal deformation.

The analysis is then extended to arbitrary \emph{piecewise differentiable} evolutions and to the associated class of infinitesimal deformations. The admissible
evolutions are classified into \emph{ordinary\/}, if every admissible infinitesimal deformation vanishing at the endpoints is tangent to some finite
deformations with fixed endpoints, and \emph{exceptional\/} in the opposite case.

Along the same guidelines, every piecewise differentiable evolution is assigned a corresponding \emph{abnormality index\/}, extending and expressing in
geometrical terms the traditional attributes of normality and abnormality commonly found in the literature for regular evolutions.

A noteworthy result, established in Appendix A, is the existence of a strict relationship between abnormality index and ordinariness, leading to the conclusion
that all \emph{normal\/} evolutions are automatically ordinary. This extends to arbitrary non--linear constraints and to piecewise differentiable sections a
result proved by Hsu \cite{Hsu} in a linear context.

After these preliminaries, in Sec. \!3, attention is focussed on the determination of the broken extremals --- or extremaloids --- of the action functional.
Once again, a fully covariant approach is proposed, yielding back Pontryagin's equations \cite{Giaquinta,Pontryagin} as well as the Erdmann--Weierstrass corner
conditions \cite{Giaquinta,Sagan}.

The resulting equations are shown to be sufficient for \emph{any\/} evolution, and  necessary and sufficient for an \emph{ordinary\/} evolution to be an
extremaloid. The same setup is seen to provide a concise approach to the Lagrange multipliers method.

A deeper insight into the nature of the problem is then gained lifting the algorithm to a more natural environment, here called the \emph{contact bundle\/},
gluing together the velocity space and the phase space. This allows recovering the ordinary extremaloids of the original variational problem as {\em
projections\/} of the extremaloids of a free variational problem on the contact bundle.

The same environment is also seen to provide an intrinsic characterization of the abnormal evolutions in terms of the geometric properties of the contact
bundle,  as well as a proof of the fact that, under suitable regularity assumptions, the original constrained variational problem is locally equivalent to a
free hamiltonian problem in phase space.

\section{Geometric setup}\label{Sec2}
In this Section we present a brief review of the geometric tools involved in the subsequent discussion. Throughout the paper, we shall freely use the language
and methods of  differential geometry \cite{Sternberg,Warner}. The terminology will be partly borrowed from classical non--holonomic Mechanics
\cite{MP2,Saunders, Pom,DeLeon}.

\subsection{Preliminaries}\label{Sec2.1}
Let $\@\V\rarw{t}\R\@$ denote a fibre bundle over the real line, henceforth called the \emph{event space\/}, and referred to local fibred coordinates
$\@t,q^1,\dots,q^n$\vspace{1pt}.

Every section $\@\g\colon\R\to\V\,$ is interpreted as the evolution, parameterized in terms of the independent variable $\/t\/$, of an abstract system $\@\B\/$
with a finite number of degrees of freedom. The first jet--bundle $\@\j\V\rarw{\pi}\V\@$, referred to local jet--coordinates $\@t,q^i,\dot q^i\/$, is called
the \emph{velocity space\/}. The first jet--extension of $\@\g\?$ is denoted by $\@\j\g\colon\R\to\j\V\@$\vspace{1pt}.

The presence of differentiable constraints is accounted for by a commutative diagram of the form
\begin{equation}\label{2.1}
\begin{CD}
\A          @>i>>       \j\V            \\
@V{\pi}VV                @VV{\pi}V      \\
\V          @=          \V
\end{CD}
\end{equation}
where the submanifold $\@\A\rarw{i\;}\j\V\,$, fibred over $\@\V\@$, represents the totality of \emph{admissible velocities\/}.
The manifold $\@\A\@$ is referred to local fibred coordinates $\@t, q^1,\ldots, q^n, z^1, \ldots, z^r\@$ with transformation laws
\begin{equation*}
\bar{t} = t + c\,,\quad\;\bar{q}\@^i=\bar{q}\@^i\/(t,q^1,\ldots,q^n)\,,\quad\;\bar{z}\@^A=\bar{z}\@^A\/(t,q^1,\ldots,q^n,z^1,\ldots,z^r)\,,
\end{equation*}
while the imbedding $\@i: \A \to j_1(\V)\@$ is locally expressed as
\begin{equation}\label{2.2}
\dot q\@^i = \psi\?^i\/(t,q^1,\ldots,q^n,z^1,\ldots,z^r) \qquad i=1,\ldots,n\@,
\end{equation}
with $\@\rank\Big\|\de\?(\psi^1\@\cdots\,\psi^n)/de{\?(z^1\@\cdots\,z^r)}\Big\|=r\@$.\vspace{1pt}
To make it easier, no distinction is made between $\@\A\@$ and its image\linebreak $\@i(\A)~\subset~j_1(\V)\@$.

\smallskip
An evolution $\@\g:\R\to\V\@$ is called \emph{admissible} if and only if its first jet extension $\@\j\g\@$ factors through $\@\A\@$, i.e.~if and only if there
exists a section
$\@\h\g:\R\to\A\,$ satisfying $\@\j\g=i\cdot\h\g\@$. In the stated circumstance, the section $\@\h\g\@$, locally described as $\@q^i =q^i\/(t),
z^A=z^A\/(t)\,$, is itself called an admissible section of $\@\A\@$, the admissibility requirement taking the explicit form
\begin{equation}\label{2.3}
\d q^i/d t \,=\, \psi^i\left(t,q^1\/(t),\ldots,q^n\/(t),z^1\/(t),\ldots,z^r\/(t)\right)\@.\vspace{2pt}
\end{equation}
With this terminology, the admissible sections of $\@\A\@$ are in bijective correspondence with the admissible evolutions of the system. We shall refer to
$\h\g\@$ as to the \emph{lift\/} of $\@\g\@$.

On account of Eq.~(\ref{2.3}), the admissible evolutions of the system are determined, up to initial data, by the knowledge of the functions $\?z^A\/(t)\@$.
In general, however, these functions have no invariant geometrical meaning\footnote%
{\@An obvious exception occurs whenever the manifold $\@\A\@$ splits into a fibred product $\@\V\times_\R\Z\,$, the second factor being some fiber bundle over
$\@\R\@$.}.

Perfectly significant tools for the description of possible evolutions of the system are instead the \emph{sections\/} $\@\s:\V\to\A\,$. Each of them is called
a \emph{control\/} for the system; the composite map $\@i\cdot\s:\V\to\j\V\/$ is called an \emph{admissible velocity field\/}.
In local coordinates we have the representations
\begin{equation*}
\begin{alignedat}{3}
&\,\s\;&&:\quad z^A\,&&=\,z^A\/(\?t\@,q^1,\ldots,q^n\?)\@, \\[1pt]
&i\cdot\s&&:\quad \dot q\@^i&&=\, \psi\@^i\/(\?t\@,q^1,\ldots,q^n,z^A\/(\?t\@,q^1,\ldots,q^n\?)\/)\@,
\end{alignedat}
\end{equation*}
pointing out that the knowledge of $\/\s\@$ does indeed determine the evolution $\@\g\/(t)\@$ from a given initial configuration $\@\g\/(t_0)\/$ through a
well--posed Cauchy problem.

A control $\/\s:\V\to\A\@$ is said to \emph{contain\/} a section $\@\g:\R\to\V\,$ if and only if the lift $\@\h\g:\R\to\A\@$\linebreak 
factors into $\@\h\g=\s\cdot\g\@$,
i.e.~if and only if the jet extension $\j\g\@$ coincides with the composite map \linebreak
$\,i\cdot\s\cdot\g:\R\to\j\V\@$.

\smallskip
\subsection{Geometry of the velocity space}\label{Sec2.2}
\textbf{(\/i\/)} \,Given the event space $\@\V\@$, we denote by $\@V\/(\V)\@$\vspace{.4pt} the vertical bundle associated with the fibration $\@\V\to\R\@$ and
by $\@V^*\/(\V)\@$ the corresponding dual bundle. The elements of $V^*\/(\V)\@$ are called the \emph{virtual $1$--forms\/} over $\V\@$.

By definition, $\@V^*\/(\V)\@$ is canonically isomorphic to the quotient of the cotangent bundle $T^*\/(\V)$ by the equivalence relation
\begin{equation}\label{2.4}
\s\sim\s'\;\Longleftrightarrow\;\left\{
\begin{aligned}
 & \pi\/(\s)=\pi\/(\s')\\
 & \s-\s'\;\propto\; d\/t_{\,|\pi\/(\s)}
\end{aligned}
\right.\@.
\end{equation}

\vskip1pt
The quotient map $T^*\/(\V)\to V^*\/(\V)$ is denoted by $\@\varpi\@$. For each function  $g\in\F\?(\V)$, the section $\delta\?g\colon\V\to V^*\/(\V)$ given by
$\delta\?g_{|x}:=\varpi\/(d\/g_{|x})\@$ is called the \emph{virtual differential\/} of $\@g\@$.

Every local coordinate system $t,q^i\/$ in $\V\@$ induces fibred coordinates $t,q^i,p_{\?i}\@$ in $\@V^*\/(\V)\,$, uniquely defined by the requirement
$\@\h\l=p_i\/(\h\l)\@\?\delta\?q^i_{\@|x}\;\,\forall\,\h\l\in V_x^*(\V)\@$.

Every element belonging to the tensor algebra generated by $\@V\/(\V)\@$ and $\@V^*\/(\V)\@$ is called a \emph{virtual tensor\/} over $\@\V\@$. We preserve the
notation $\@\langle\;\;,\;\,\rangle\@$ for the pairing between $\@V(\V)\@$ and $\@V^*(\V)\@$.

\medskip \noindent
\textbf{(\/ii\/)} \,The fibred product $\@\C\?(\A):=\A\times_{\V} V^*\/(\V)$ is called the {\it contact bundle\/}: it is at the same time a vector bundle over
$\/\A\@$, isomorphic to the subbundle of $\@T^*\/(\A)\@$ locally spanned by the $1$--forms $\@\w^i:=d\/q^i-\psi^i\/d\/t\@$, and a fibre bundle over
$\@V^*\/(\V)\@$.\vspace{-1pt} The situation is summarized into the commutative diagram
\begin{equation}\label{2.5}
\begin{CD}
\C\/(\A )         @>\k>>          V^*\/(\V)        \\
@V{\z}VV                          @VV\pi V         \\
\A                @>\pi>>           \V
\end{CD}
\end{equation}

We refer $\@\C\?(\A)\?$ to coordinates $\@t,q^i,z^A,p_{\?i}\@$ gluing in an obvious way the coordinates in the factor spaces. Every~$\@\s\in\C\/(\A)\@$ is
called a \emph{contact $1$--form\/} over $\@\A\@$.

As a subbundle of the cotangent space $\@T^*\/(\A)\@$, the manifold $\@\C\/(\A)\@$ is naturally endowed with a \emph{Liouville $1$--form\/} $\@\Th\@$, uniquely
defined by the relation
\begin{equation*}
\big<\@Z\@,\@\Th\@\big>\,=\,\big<(\z_{\vert \s})_*\/(Z)\@,\@\s\@\big>\qquad\;\forall\;Z\in T_\s\/(\C\/(\A))
\end{equation*}
and expressed in coordinates as
\footnote%
{For simplicity, we do not distinguish between covariant objects in $\@\A\@$ and their pull--back in $\@\C\/(\A)\@$; namely, we write $\@\psi^i\@$ for
$\@\z^*\/(\psi^i)\@,\,\@\w^i\@$ for $\@\z^*\/(\w^i)\@$ etc.}
\begin{equation}\label{2.6}
\Th\,=\,p_{\?i}\,\?\w^i\,=\,p_{\?i}\big(\/d\/q^{\?i}-\@\psi^i\@d\/t\/\big)\@.
\end{equation}

\smallskip\noindent
\textbf{(\/iii\/)} \,A vector field $\@\G\@$ over $\@\A\@$ is called a \emph{semispray\/} if and only if its integral curves are jet--extensions of admissible
sections $\@\g:\R\to\V\@$. In coordinates, this property is summarized into the representation $\@\G=\de/de t\@+\@\psi^{\@i}\@\de/de{q^i}\@
+\@\G^A\@\de/de{z^A}\@$.\vspace{1pt}

According to the latter, given any covariant tensor field $\@w\@$ on $\@\V\@$, the Lie derivative $\@\L\@_\G\/(\pi^*\/(w))\@$ is a covariant field on $\@\A\@$,
depending only on $\@w\@$ and \emph{not\/} on the choice of the semispray~$\G$.

The correspondence $\@w\mapsto\L\@_\G\/(\pi^*\/(w))\@$, essentially identical to Tulczyjew's operator $d_T$ \cite{Tulczyjew1,Tulczyjew3,Tulczyjew4}, is here
denoted by $\@\d/dt\@$ and called the \emph{symbolic time derivative\/}: by construction, it is a derivation of the covariant tensor algebra over $\@\V\@$ with
values in the covariant algebra over $\@\A\@$, uniquely characterized by the properties
\begin{subequations}\label{2.7}
\begin{align}
  &\d\/f/dt\,=\,\L\@_\G\/\big(\pi^*\/(f)\big)\,=\,\de f/de t +\psi^k\@\de f/de{q^k}\,:=\,\dot f\@,                \\[3pt]
  &\d/dt\big(d\?f\big)\,=\,\L\@_\G\/\big(\pi^*\/(d\?f)\big)\,=\,d\@\L\@_\G\big(\pi^*\/(f)\big)\,=\,d\?\dot f\@.
\end{align}
\end{subequations}

\medskip\noindent
\textbf{(\/iv\/)} \,The vertical bundle associated with the fibration $\@\A\rarw{\pi}\V\@$\vspace{.75pt} is denoted by $\@V\/(\A)\,$. Each
$\@Y=Y^A\big(\de/de{z^A}\@\big)_{\!z}\in V_z\/(\A)\@$\vspace{1pt} induces an action on the ring of differentiable functions over $\@\V\@$, based on the
prescription
\begin{equation*}
\r\/(Y)\/(f):=Y\/(\dot f)\,=\,Y^A\biggl(\de\psi^i/de{z^A}\@\biggr)_{\!z}\@\biggl(\de f/de{q^i}\biggr)_{\pi\/(z)}\qquad\forall\,f\in \F\/(\V)\@.
\end{equation*}

According to the latter, $\@\r\/(Y)\@$ is a derivation of $\@\F\/(\V)\@$, i.e.~a (vertical) vector on $\@\V\@$. We have therefore a fibred homomorphism
\begin{equation*}
\begin{CD}
V\/(\A)         @>\r>>            V\/(\V)          \\
@VVV                               @VVV            \\
\A              @>\pi>>             \V
\end{CD}
\end{equation*}
\\[1pt]
expressed in components as
\begin{equation}\label{2.8}
\r\@\biggl(\de/de{z^A}\@\biggr)_{\!z}\@=\@\biggl(\de\psi^i/de{z^A}\@\biggr)_{\!z}\@\biggl(\de/de{q^i}\biggr)_{\pi\/(z)}\@.
\end{equation}

\smallskip
\subsection{Vector bundles along sections}\label{Sec2.3}
\textbf{(\/i\/)} \,Given an admissible section $\@\g\colon\R\to\V\@$, we denote by $\@\Vg\rarw{t}\R\@$ the bundle of vertical vectors along $\@\g\@$, by
$\@\Wg\rarw{t}\R\@$ the restriction of $\@V\/(\A)\,$ to the lift $\@\h\g\colon\R\to\A\@$, and by $\@\Ag\rarw{t}\R\@$ the bundle of \emph{isochronous\/} vectors
along $\@\h\g\@$, meant as vectors annihilating the $1$--form $\@d\/t\@$.

\vspace{1pt}
We adopt fibred coordinates $\?t,u^i\@$ in $\@\Vg\@$, $\?t,v^A\@$ in $\@\Wg\@$ and $\?t,u^i\!,v^A\@$ in $\@\Ag\@$, according to the prescriptions
\begin{subequations}\label{2.9}
\begin{alignat}{2}
 & X\,=\,u^i\/(X)\bigg(\de/de{q^i}\bigg)_{\!\g\/(t\/(X))}         &&\forall\;X\in\Vg\,,\hskip.6cm                                   \\[4pt]
 & Y\,=\,v^A\/(Y)\bigg(\de/de{z^A}\bigg)_{\!\h\g\/(t\/(Y))}       &&\forall\;Y\in \Wg\,,\hskip.6cm                                  \\[4pt]
 & \h X\,=\,u^i\/(\h X)\bigg(\de/de{q^i}\bigg)_{\!\h\g\/(t\/(\h X))}\!\!+\,v^A\/(\h X)\bigg(\de/de{z^A}\bigg)_{\!\h\g\/(t\/(\h X))}
\qquad\;&&\forall\;\h X\in\Ag\,.
\end{alignat}
\end{subequations}

The symbolic time derivative defined in Sec.\! 2.2 induces a derivation --- still denoted by $\@\d/dt\@$ --- of the algebra of covariant tensor fields along
$\@\g\@$ into the algebra of covariant tensor fields along $\@\h\g\@$.

On account of Eqs.\;(\ref{2.7}\@a, b), the effect of $\@\d/dt\@$ on a function $\@f\/(t)\@$ coincides with the ordinary derivative, while on a $1$--form
$\@w=w_0\/(t)\@(d\/t)_{\g}+w_i\/(t)\@(d\/q^i)_{\g}\@$ it reads
\begin{equation}\label{2.10}
\d w/dt\,=\,\biggl[\d w_0/dt\,d\/t\,+\,\d w_i/dt\@d\/q^i\,+\,w_i\@\biggl(\de\psi^i/de t\,d\/t\,+\,\de\psi^i/de{q^k}\,d\/q^k\,+\,
\de\psi^i/de{z^A}\,d\/z^A\biggr)\@\biggr]_{\h\g}\@.
\end{equation}

In a similar way, the correspondence (\ref{2.8}) determines an injective homomorphism $\@\Wg\xrightarrow{\h\r}\Vg\@$, expressed in components as
\begin{equation}\label{2.11}
\h\r\@\bigg(\de/de{z^A}\bigg)_{\!\h\g}\,=\@\bigg(\de\psi\?^i/de{z^A}\bigg)_{\!\h\g}\@\bigg(\de/de{q^i}\bigg)_{\!\g}\@.
\end{equation}

\smallskip\noindent
\textbf{(\/ii\/)} \,A straightforward argument, left to the reader, ensures that the first jet--bundle $\@\j\Vg\@$ is canonically diffeomorphic to the space of
isochronous vectors along $\@\j\g\@$. In jet coordinates, this results into the identification
\begin{equation*}
Z\,=\,u^i\/(Z)\bigg(\de/de{q^i}\bigg)_{\!\j\g\/(t\/(Z))}+\;\dot u^i\/(Z)\bigg(\de/de{\dot q^i}\bigg)_{\!\j\g\/(t\/(Z))}
\qquad\;\forall\;Z\in\j\Vg\,.
\end{equation*}

Due to the latter, the push--forward of the imbedding (\ref{2.2}), restricted to the submanifold $\@\Ag\subset T\/(\A)\@$, determines a vector bundle
homomorphism
\begin{subequations}\label{2.12}
\begin{equation}
\begin{CD}
\Ag          @>i_*>>       \j\Vg            \\
@V{\pi_*}VV                @VV{\pi_*}V      \\
\Vg\;\;      @=             \Vg
\end{CD}\vspace{5pt}
\end{equation}
expressed in coordinates as\vspace{5pt}
\begin{equation}
\dot u^i=\left(\de\psi^i/de{q^k}\right)_{\!\h\g}\@u^k + \left(\de\psi^i/de{z^A}\right)_{\!\h\g}\@v^A\@.
\end{equation}
\end{subequations}
The kernel of the projection $\@\Ag\rarw{\pi_*}\Vg\@$ is easily recognized to coincide with the vertical bundle $\@\Wg\@$.

\medskip\noindent
\textbf{(\/iii\/)} \,The restriction of the space $\@V^*\/(\V)\@$ to the curve $\@\g\@$ determines a vector bundle $\@V^*\/(\g)\rarw{t}\R\@$, dual to the
vertical bundle $\@\Vg\@$.

\vspace{1pt}
The elements of $V^*\/(\g)\@$ are called the \emph{virtual $1$--forms\/} along $\g\@$. We recall that, by definition, they are not $1$--forms in the ordinary
sense, but equivalence classes of $\@1$--forms under the relation (\ref{2.4}).

The elements of the tensor algebra generated by $\@\Vg\@$ and $\@V^*\/(\g)\@$ are called the \emph{virtual tensors\/} 
\linebreak 
along $\@\g\@$.

Preserving the notation $\@\varpi:T^*_\g\/(\V)\to V^*\/(\g)\@$ for the quotient map sending each (ordinary) $1$--form $\@\l\@$ along $\@\g\@$ into the
corresponding equivalence class and $\@\delta\?q^i$~for the image $\@\varpi\/(d\/q^i)\@$, every virtual tensor field $\@w\colon\R\to\Vg\otimes_\R
V^*\/(\g)\otimes_\R\cdots\,$ is locally represented as
\begin{equation*}
w\@=\,w^{\@i}{}_{j\@\cdots}\/(t)\,\bigg(\de/de{q^i}\bigg)_{\!\g}\!\otimes\wg j\otimes\cdots\@.
\end{equation*}

Every local coordinate system $t,q^i\/$ in $\V\@$\vspace{1pt} induces fibred coordinates $t,p_i\@$ in $\@V^*\/(\g)\@$, with 
\linebreak
$\@p_i\/(\h\l):=\big<\big(\de/de{q^i}\big)_\g\@,\@\h\l\big>\@$ $\,\forall\,\h\l\in V^*\/(\g)\@$.

\smallskip
\subsection{Deformations}\label{Sec2.4}
A deformation $\g\?_\xi\colon\R\to\V\@$ of an admissible section $\@\g\@$ --- and, likewise, a deformation $\@\h\g\?_\xi\colon\R\to\A\@$ of the lift $\@\h\g\@$
--- are called \emph{admissible\/} if and only if all sections $\@\g\?_\xi\@,\@\h\g\?_\xi\@$ are admissible, i.e.~if and only if they satisfy the condition
$\j{\g\?_\xi}=i\cdot\h\g\?_\xi\@$.

In coordinates, assuming the representation $\@\h\g:\@q^i=q^i\/(t)\,,\,z^A=z^A\/(t)\,$, the admissible deformations of $\?\h\g\@$ are described by equations of
the form
\begin{equation*}
\h\g\?_\xi\@:\qquad q^i=\vphi^i\/(\xi,t)\,,\quad z^A=\z^A\/(\xi,t)\@,
\end{equation*}
subject to the conditions
\begin{subequations}\label{2.13}
\begin{align}
 & \vphi^i\/(0,t)\@=\@q^i\/(t)\,,\quad \z^A\/(0,t)\@=\@z^A\/(t)\@,                  \\[3pt]
 &\de\vphi^i/de t\,=\,\psi^i\/\big(t,\vphi^i\/(\xi,t),\z^A\/(\xi,t)\big)\@.
\end{align}
\end{subequations}

\smallskip
Setting $\@X^i\/(t):=\big(\de\@\vphi^i/de{\?\xi\;}\big)_{\xi=0}\,\@,\,X^A\/(t):=\big(\de\@\z\/^A/de{\?\xi\;}\big)_{\xi=0}\,$\vspace{1pt}, the infinitesimal
deformation tangent to $\@\h\g\?_\xi\,$ is the section $\@\h X\colon\R\to\A\?(\h\g)\@$ locally expressed as
\begin{equation*}
\h X\,=\,X^i\/(t)\left(\de/de{q^i}\right)_{\!\h\g}\,+\,X^A\/(t)\left(\de/de{z^A}\right)_{\!\h\g}\@,
\end{equation*}
while the admissibility condition (\ref{2.13}\@b) is reflected into the \emph{variational equation\/}
\begin{equation}\label{2.14}
\d X^i/d t\,=\,\left(\de\@\psi^i/de{q^k}\right)_{\!\h\g}X^k\,+\,\left(\de\@\psi^i/de{z^A}\right)_{\!\h\g}X^A\@.
\end{equation}

The infinitesimal deformation tangent to the projection $\@\g\?_\xi=\pi\cdot\h\g\?_\xi\,$ is similarly identified with the section $\@X\colon\R\to\Vg\@$
locally expressed by
\begin{equation*}
X=\pi_*\@\h X\,=\,X^i\/(t)\@\left(\de/de{q^i}\right)_{\!\h\g}\@.
\end{equation*}
Collecting all results and recalling Eq.\;(\ref{2.12}b) we conclude
\begin{proposition}\label{Pro2.1}
Let $\@\h\g:\R\to\A\,$ be the lift of an admissible evolution $\@\g:\R\to\V$. Then, a section $\?X:\R\to\Vg\?$ represents an admissible infinitesimal
deformation of $\@\g\@$ if and only if its first jet extension factors through $\@\Ag\@$, i.e.~if and only if there exists a section $\@\h X :\R\to\Ag\,$
satisfying $\@\j X=i_*\/\h X\/$.
\\
Conversely, a section $\@\h X:\R\to\Ag\@$ represents an admissible infinitesimal deformation of $\@\h\g\@$ if and only if it projects into an admissible
infinitesimal deformation~of $\@\g\@$, i.e.~if and only if $\,i_*\h X=\j{\pi_*\@\h X}\@$.
\end{proposition}
The proof is entirely straightforward, and is left to the reader.

\smallskip
Proposition \ref{Pro2.1} points out the completely symmetric roles played by diagram (\ref{2.1}) in the study of the admissible \emph{evolutions\/} and by
diagram (\Ref{2.12}a) in the study of the admissible infinitesimal \emph{deformations\/}, thus enforcing the viewpoint that the latter context is essentially a
``linearized counterpart'' of the former one.

\begin{remark}\label{Rem2.1}
Throughout the previous discussion, an admissible infinitesimal deformation is meant as a vertical vector field along $\@\g\@$ (or as an isochronous vector
field along $\@\h\g\@$) tangent to an admissible finite deformation {\em whatsoever}. A necessary and sufficient condition for this to hold is the validity of
the variational equation~(\ref{2.14}).

Any additional requirement on the allowed finite deformation --- such as e.g.~keeping the endpoints fixed --- is reflected into a corresponding restriction on
the infinitesimal ones. We shall return on this point in Sec.~\!\ref{Sec2.7}.
\end{remark}

\smallskip
\subsection{Infinitesimal controls}\label{Sec2.5}
\medskip\noindent
\textbf{(\/i\/)} \,According to Proposition \ref{Pro2.1}, the admissible infinitesimal deformations of an admissible section $\@\g:\R\to\V\,$ are in 1--1
correspondence with the sections $\@\h X:\R\to\Ag\,$ satisfying the consistency requirement $\@i_*\h X=\j{\pi_*\@\h X}\@$.

In local coordinates, setting $\@\h X=X^i\/(t)\,\de/de{q^i}\@+X^A\/(t) \,\de/de{z^A}\,$, the requirement is expressed by the variational equation (\ref{2.14}).

Exactly as it happened with Eq.\;(\ref{2.3}), Eq.\;(\ref{2.14}) indicates that, for each admissible $\@\h X$, the functions $\@X^i\/(t)\/$ are determined, up
to initial data, by the knowledge of $\@X^A\/(t)\/$.
Once again, however, one has to cope with the fact that the components $\@X^A$ have no invariant geometrical meaning.

The difficulty is overcome introducing a linearized version of the idea of \emph{control\/}.
\begin{definition}\label{Def2.1}
Let $\@\g:\R\to\V\@$ denote an admissible evolution. Then:
\begin{itemize}
\item a linear section $\@h:\Vg\to\Ag\@$, meant as a vector bundle homomorphism satisfying $\@\pi_*\cdot h=id\@$, is called an \emph{infinitesimal
    control\/} along $\@\g\@$;\vspace{1pt}
\item the image $\@\Hg:=h\@(\Vg)\@$, viewed as a vector subbundle of $\@\Ag\to\R\@$, is called the \emph{horizontal distribution\/} along $\@\h\g\@$
    induced by $\@h\@$; every section $\@\h X:\R\to\Ag\@$ satisfying $\@\h X\/(t)\in\Hg\;\forall\;\@t\in\R\@$ is called a horizontal section.
\end{itemize}
\end{definition}

\begin{remark}\label{Rem2.2}\rm
The term \emph{infinitesimal control\/} is intuitively clear: given an admissible section $\?\g\?$, let $\@\s:\V\to\A\@$ denote any control containing
$\?\g\@$, i.e.~satisfying $\?\s\cdot\g=\h\g\@$. Then, on account of the identity $\?\pi_*\cdot\s_* =(\pi\cdot\s)\@_* =id\@$, the tangent map $\?\s_*:T\/(\V)\to
T\/(\A)\@$, restricted to $\?\Vg\@$, determines a linear section
\linebreak
$\?\s_*:\Vg\to\Ag\@$ satisfying the requirements of Definition \ref{Def2.1}. The infinitesimal controls may therefore be thought of as equivalence classes of
ordinary controls having a first order contact along~$\g\@$.
\end{remark}

\medskip\noindent
\textbf{(\/ii\/)} \,Given an infinitesimal control $\@h:\Vg\to\Ag\@$, the horizontal distribution $\@\Hg\@$ and the vertical subbundle $\@\Wg\@$ split the
vector bundle $\@\Ag\@$ into the fibred direct sum
\begin{equation}\label{2.15}
\Ag\,=\,\Hg\oplus_{\@\R}\/\Wg\@,
\end{equation}
thereby giving rise to a couple of homomorphisms $\@\P_\H:\Ag\to\Hg\,$ (horizontal projection) and $\@\P_V:\Ag\to\Wg\@$ (vertical projection), uniquely defined
by the relations
\begin{equation}\label{2.16}
\P_\H\@=\@h\cdot\pi_*\ ;\qquad\P_V\@=\,id\@ -\P_\H\@.
\end{equation}

In fibred coordinates, every infinitesimal control $\@h:V\/(\g)\to\Ag\,$ is locally represented as
\begin{equation}\label{2.17}
v^A\@=\@h_i{}^A\/(t)\,u^i\@.
\end{equation}
In this way:
\begin{itemize}
\item
the horizontal distribution $\@\Hg\@$ is locally spanned by the vector fields
\begin{equation}\label{2.18}
\DE_i:=h\@\bigg[\left(\de/de{q^i}\right)_{\!\g}\@\bigg]\,=\,\left(\de/de{q\?^i}\right)_{\!\h\g} + h\@_i{}^A\,\left(\de/de{z^A}\right)_{\!\h\g}\@;
\end{equation}
\item
every vertical vector field $\@X=X^i\/(t)\big(\de/de{q^i}\big)_{\!\g}\@$ along $\@\g\@$ may be lifted to a horizontal field $\@h\/(X)\@$ along $\@\h\g\@$,
expressed in components as
\begin{equation}\label{2.19}
h\/(X)\,=\,X^i\/(t)\,\DE_i\,=\,X^i\/(t)\@\bigg[\left(\de/de{q\?^i}\right)_{\!\h\g} +
h\@_i{}^A\,\left(\de/de{z^A}\right)_{\!\h\g}\@\bigg]\@;
\end{equation}
\item
every vector $\@\h X=X^i\@\big(\de/de{q\?^i}\big)_{\h\g}+ X^A\@\big(\de/de{z^A}\big)_{\h\g}\@\in\Ag\,$ admits a unique representation of the form $\@\h
X=\P_\H\/(\h X)+\P_V\/(\h X)\@$, with
\begin{equation}\label{2.20}
\P_\H\/(\h X)=X^i\,\DE_i\,,\quad\P_V\/(\h X)=\Big(X^A\!-X^i\@h\@_i{}^A\Big)\!\left(\de/de{z^A}\right)_{\!\h\g}\!\!=:
U^A\!\left(\de/de{z^A}\right)_{\!\h\g}\@;
\end{equation}
\item
every $1$--form $\@\nu\@$ along $\@\h\g\@$ determines a linear functional on $\@\A\/(\h\g)\@$ and therefore, by duality, a virtual $1$--form $\@h^*(\nu)\@$
along $\@\g\@$, uniquely defined by the relation
\begin{equation}\label{2.21}
\big<\?h^*(\nu)\@,\@X\@\big>\@=\@\big<\?\nu\@,\@h\/(X)\@\big>\;\;\forall\;X\in\Vg\quad\Longrightarrow\quad
h^*(\nu)\@=\@\big<\?\nu\@,\@\DE_i\@\big>\,\wg i\,.\hskip.3cm
\end{equation}
\end{itemize}

\medskip
The role of Definition \ref{Def2.1} is further enhanced by the following
\begin{definition}\label{Def2.2}
Let $\@h\/$ be an infinitesimal control along the (admissible) section $\@\g\@$. A section 
\linebreak
$\@X:\R\to\Vg\@$ is said to be $\@h$--transported along $\@\g\@$ if
and only if its horizontal lift $\@h\/(X):\R\to\Ag\@$ is an admissible infinitesimal deformation of $\@\h\g\@$, i.e.~if and only if $\;i_*\cdot h\/(X)=\j X\@$.
\end{definition}

In coordinates, setting $X=X^i\/(t)\@\big(\de/de{q^i}\big)_\g\plus06\@$ and recalling Eqs.\;(\ref{2.14}), (\ref{2.19}), the condition for $h$--transport is
expressed by the linear system of ordinary differential equations
\begin{equation}\label{2.22}
\d X^i/d t\@=\left[\bigg(\de\psi^i/de{q^k}\BIGG)_{\h\g}\,+
\,h_k{}^A\,\bigg(\de\psi^i/de{z^A}\BIGG)_{\h\g}\,\right]\@X^k\@ = \@X^k\,\DE_k\@\psi^i\@.
\end{equation}

From the latter, on account of Cauchy theorem, we conclude that the $h$--transported sections of $\@\Vg\@$ form an $n$--dimensional vector space $\@V_h\@$,
isomorphic to each fibre $\@\Vg_{|t}\,$ through the evaluation map $\@X\to X\/(t)\,$. We can therefore state:
\begin{proposition}\label{Pro2.2}
Every infinitesimal control $\@h:\Vg\to\Ag\@$ determines a trivialization of the vector bundle $\@\Vg\xrightarrow{t\,}\R\@$.
\end{proposition}

Proposition \ref{Pro2.2} provides an identification between sections $\@X:\R\to\Vg\@$ and vector valued functions $\@X:\R\to V_h\@$ --- whence, by duality,
also an identifica\-tion between sections $\@\h\l:\R\to V^*\/(\g)\@$ and vector valued functions $\@\h\l:\R\to V^*_h\@$ --- thus opening the door to the
introduction of an \emph{absolute time derivative\/} $\D/D t\@$ for virtual tensor fields along $\@\g\@$.
The algorithm is readily implemented in components.\vspace{1pt}

To this end, let $\big\{e_{(a)}\big\}$, $\big\{e^{(a)}\big\}$ be any pair of dual bases for the spaces $\?V_h$, $\?V_h^*$.\vspace{1pt} By definition, each
$e_{(a)}$ is a vertical vector field along $\g$, obeying the transport law~(\ref{2.22}). In coordinates, setting
$\,e_{(a)}=\e_a^i\,\big(\de/de{q^i}\big)_\g\,$, this implies the relation
\begin{subequations}\label{2.23}
\begin{equation}
\d \e_a^i/d t\,=\,\e_a^k\,\@\DE_k\@\?\psi^i\@.
\end{equation}

In a similar way, each $\@e^{(a)}\@$ is a virtual $1$--form along $\@\g\@$\vspace{1pt}, expressed on the basis $\@\wg i\@$ as $\,e^{(a)}=\e^a_i\,\wg i\@$, with
$\,\e^a_i\,\e_b^i\@=\@\delta\@^a_b\,$\vspace{-2pt}. On account of Eq.\;(\Ref{2.23}a), the components $\@\e^a_i\@$ obey the transport law
\begin{equation}
\d/d t\,\Big(\@\e^a_i\,\e_a^{\?j}\@\Big)\,=\,0\qquad\Longrightarrow\qquad \d \e^a_i/d t\,=\,-\,\e^a_j\,\DE_i\@\psi\@^j\vspace{8pt}\@.
\end{equation}
\end{subequations}

The functions\vspace{-4pt}
\begin{subequations}\label{2.24}
\begin{equation}
\tau\?_i\@\?^j\,:=\,\d \e^a_i/d t\;\e_a^{\?j}\,=\,-\,\e^a_i\;\d\e_a^{\?j}/d t\vspace{.5em}
\end{equation}
are called the {\it temporal connection coefficients\/} associated with the infinitesimal control $\@h\@$ in the coordinate system $\@t\?,q^i\@$.
Eqs.~(\ref{2.18}), (\Ref{2.23}a,b) provide the expression
\begin{equation}
\tau\?_i\@^j\,=\,-\,\DE_i\@\psi\@^j\,=\,-\,\left(\de\psi^j/de{q\?^i}\right)_{\!\h\g}
\@-\@h\@_i{}^A\,\left(\de\psi^j/de{z^A}\right)_{\!\h\g}\@.
\end{equation}
\end{subequations}

Given any section $\@X:\R\to\Vg\@$, the definition of the operator $\@\D/Dt\,$ yields the evaluation
\begin{equation*}
\D X/D t\@=\@\d/dt\,\Big<\@X,\@e^{(a)}\@\Big>\,e_{(a)}\,=\@\d/dt\,\Big(\@X^i\@\e^a_i\@\Big)\,\e_a^j\,\bigg(\de/de{q^j}\bigg)_{\!\g}\@,
\end{equation*}
written more simply as
\begin{subequations}\label{2.25}
\begin{equation}
\D X/D t\,=\left(\@\d\/X^j/d t\@+\@X^i\@\tau\?_i\@\?^j\right)\!\bigg(\de/de{q^j}\bigg)_{\!\g}\@.
\end{equation}

In a similar way, for any virtual $1$--form $\@\h\l:\R\to V^*\/(\g)\@$, the same argument provides the result
\begin{equation}
\D\h\l/D t\@=\@\d/dt\,\Big<\@\h\l\@,\@e_{(a)}\@\Big>\,e^{(a)}=\d/dt\@\Big(\@\l_i\,\e_a^i\@\Big)\@\e^a_j\,\wg j=
\bigg(\@\d\/\l_j/d t\@-\@\l_i\,\tau\?_j{}^i\bigg)\wg j\@.
\end{equation}
\end{subequations}

\medskip\noindent
\textbf{(\/iii\/)} \,In view of Eq.\;(\ref{2.20}), every infinitesimal deformation $\@\h X\@$\vspace{1pt} of the section $\@\h\g\@$ admits a unique
representation of the form $\@\h X=h\/(X)+U$, where:
\begin{itemize}
\item
$\@X\@=\@\pi_*\@\h X\@:=\@X^i\@\big(\de/de{q^i}\big)_\g\@$ is a vertical field along $\@\g\@$, namely the (unique) infinitesimal deformation of $\@\g\@$
lifting to $\@\h X\@$;
\smallskip
\item
$\@U\@\@=\@\P_V\/(\h X)\@:=\@U^A\@\big(\de/de{z^A}\big)_{\h\g}\@$ is a vertical vector field along $\@\h\g\@$.
\end{itemize}

In terms of this decomposition, the variational equation (\ref{2.14}) reads
\begin{equation*}
\d X^i/dt\,=\,X^k\bigg(\de\psi^i/de{q^k}\bigg)_{\!\h\g}+\,
\bigg(\de\psi^i/de{z^A}\bigg)_{\!\h\g}\left(X^k\@h_k{}^A\@+\@U^A\right)\@.
\end{equation*}

Recalling Eqs.\;(\Ref{2.24}b), (\Ref{2.25}a), as well as the representation (\Ref{2.11}) for the homomorphism
\linebreak
$\@\h\r:\Wg\to\Vg\@$, the latter equation is more
conveniently written as
\begin{subequations}\label{2.26}
\begin{equation}
\D X/D t\,=\,\h\r\@\big(\@U\@\big)=\,\h\r\@\big(\@\P_V\/(\h X)\@\big)
\end{equation}
or also, setting $\@X=X^a\@e_{(a)}\@$, $\@U=U^A\@\big(\de/de{z^A}\big)_{\h\g}\@$ and expressing everything in the $h$--transported basis $\@e_{(a)}$
\begin{equation}
\d X^a/d t=\Big<\@e^{(a)}\@,\,\h\r\big(U\big)\@\Big>\,=\,
\e^a_i\;\bigg(\de\psi^i/de{z^A}\BIGG)_{\h\g}\;U^A\@.
\end{equation}
\end{subequations}

Precisely as the original equation (\ref{2.14}), Eq.\;(\Ref{2.26}a) points out that every infinitesimal deformation $\@X\@$ is determined by the knowledge of a
vertical vector field $\@U:\R\to\Wg\@$ through the solution of a well posed Cauchy problem:\vspace{1pt} the advantage of the present formulation is that all
quantities have now a precise geometrical meaning relative to the horizontal distribution $\@\Hg\,$ induced by the infinitesimal control $\@h\@$.

As we shall see, this aspect has useful consequences. At the same time, of course, one cannot overlook the fact that, in the formulation of a typical
variational problem, no distinguished section 
\linebreak
$\?h:\Vg\to\Ag\@$ is generally included among the data, and no one is needed in order to formulate the results.

In this respect, the infinitesimal controls play the role of \emph{gauge fields\/}, useful for covariance purposes, but unaffecting the evaluation of the
extremals. Accordingly, in the subsequent analysis we shall regard $h$ as a user--defined object, eventually checking the invariance of the results under
arbitrary changes~$h\to h'$.

\smallskip
\subsection{Corners}\label{Sec2.6}
As anticipated in the Introduction, in the study of the extremals of a variational problem we shall not only consider \emph{ordinary\/} sections, but also
\emph{piecewise differentiable} ones, defined on \emph{closed} intervals. To this end, we stick to the following standard terminology:
\begin{itemize}
 \item an admissible closed arc $\@\arc\g{\?a,b\?}\@$ in $\@\V\@$ is the restriction to a closed interval $\@[\?a,b\@]\@$ of an admissible section
$\@\g:(c\@,d\@)\to\V\@$ defined on some open interval $\@(c\@,d\@)\supset[\?a,b\@]\@$;
\vspace{1pt}
\item a piecewise differentiable evolution of the system in the interval $\@[t_0\?,\?t_1]\@$ is a finite collection
\begin{equation*}
\arc{\?\g}{\?t_0\?,\?t_1}:=\big\{\arc{\?\g^{\narc{s}}}{\?a_{s-1},a_s}\?,\,s=1,\ldots,N,\;t_0=a_0<a_1<\cdots<a_N=t_1\/\big\}
\end{equation*}
of admissible closed arcs satisfying the matching conditions
\begin{equation}\label{2.27}
\g^{\ns}\/(a_s)\,=\,\g^{\nS}\/(a_s)\qquad\;\forall\,s=1,\ldots,N-1\@.
\end{equation}
\end{itemize}

According to Eq.\;(\ref{2.27}), the image $\@\g\/(t)\@$ is well defined and continuous for all $\?t_0\le t\le t_1\@$. This allows to regard the map
$\@\g:[\@t_0,t_1]\to\V\@$ as a section in a broad sense. The points $\@\g\/(t_0)\@,\,\g\/(t_1)\,$ are called the endpoints of $\@\g\@$. The points
$\@x_s:=\g\/(a_s)\,$, $\,s=1,\ldots,N-1\@$ are called the \emph{corners\/}.

Consistently with the stated definitions, the lift of an admissible closed arc $\@\arc{\g}{\?a,b\@}\@$ is the restriction to $\@[\?a,b\@]\@$ of the lift
$\@\h\g:(c\@,d\@)\to\A\@$, while the lift $\arc{\?\h\g}{\?t_0\?,\?t_1}$ of a piecewise differentiable evolution
$\@\big\{\arc{\@\g^{\ns}}{\@a_{s-1},a_s}\@\big\}\,$ is the family of lifts $\@\h\g^{\ns}$, each restricted to the interval $\@[\@a_{s-1},a_s\?]\@$.

With this definition, the image $\@\h\g\/(t)\@$ is well defined for all $\?t\ne a_1,\ldots,a_{N-1}\@$, thus allowing to regard $\@\h\g:[\@t_0,t_1]\to\A\@$ as a
(generally discontinuous) section of the velocity space. In particular, due to the fact that the map $\@i:\A\to\j\V\@$ is an imbedding of $\@\A\@$ into an
\emph{affine\/} bundle over $\@\V\@$, each difference
\begin{equation*}
\big[\?\h\g\@\big]_{a_s}=\@i\@\big(\@\h\g^{\nS}\/(a_s)\@\big)-\@i\@\big(\@\h\g^{\ns}\/(a_s)\@\big)\@,\;s=1,\ldots,N-1
\end{equation*}
identifies a vertical vector in $\@T_{x_s}\/(\V)\@$, called the \emph{jump\/} of $\@\h\g\@$ at the corner $\@x_s\,$.

In local coordinates, setting $\@q^{\;i}_{\ns}\/(t):=q^i\/(\g^{\ns}\/(t))\@$ and introducing the notation
\linebreak
$\big[\?\psi^i\/(\h\g)\?\big]_{\?a_s}:=\?\psi^i\?(\h\g^{\?\nS}\/(a_s))- \psi^i\?(\h\g^{\?\ns}\/(a_s))\/$ to indicate the jump of the function
$\psi^i\/(\?\h\g\/(t))$ at $t=a_s\@$, Eqs.\;(\ref{2.3}), (\ref{2.27}) provide the representation
\begin{equation}\label{2.28}
\big[\?\h\g\@\big]_{a_s}=\bigg(\bigg(\d q^{\;i}_{\nS}/dt\bigg)_{\!a_s}-
\bigg(\d q^{\;i}_{\ns}/dt\bigg)_{\!a_s}\bigg)\bigg(\de/de{q^i}\bigg)_{\!x_s}
\@=\Big[\?\psi^i\/(\h\g)\Big]_{a_s}\bigg(\de/de{q^i}\bigg)_{\!x_s}\@.
\end{equation}

\smallskip
Along the same guidelines, an admissible deformation of an admissible closed arc $\arc\g{\?a,b\?}$ is a 
\linebreak
$1$--parameter family
$\arc{\g_\xi}{\?a\/(\xi),b\/(\xi)\@}$, $\@|\?\xi\?|<\eps$ of admissible closed arcs depending differentiably on $\@\xi\@$ and satisfying
$\@\arc{\g_0}{\?a\/(0),b\/(0)\@}=\arc\g{\?a,b\?}\@$. Notice that the definition explicitly includes possible variations of the reference interval
$\@[\?a\/(\xi),b\/(\xi)\@]\@$.

In a similar way, an admissible deformation of a piecewise differentiable evolution $\@\arc\g{\?t_0,t_1\?}\@$ is a collection
$\@\big\{\arc{\@\g^{\ns}_{\;\xi}}{\@a_{s-1}\/(\xi),a_s\/(\xi)\@}\/\big\}\@$ of deformations of the various arcs, satisfying the matching conditions
\begin{equation}\label{2.29}
\g^{\ns}_{\;\xi}\/(a_s\/(\xi))\,=\,\g^{\nS}_{\;\xi}\/(a_s\/(\xi))\qquad\;
\forall\;|\?\xi\?|<\eps\,,\;s=1,\ldots,N-1\@.
\end{equation}

In the stated circumstance, the lifts $\@\h\g_\xi\@$ and $\@\h\g^{\ns}_{\;\xi}\@$, respectively restricted to the intervals $\@ [\?a\/(\xi),b\/(\xi)\@]\,$ and
$\@[\@a_{s-1}\/(\xi),a_s\/(\xi)\,)\,$ are easily recognized to provide deformations for the lifts $\@\h\g:[\?a,b\@]\to\A\,$ and
$\@\h\g^{\ns}:[\@a_{s-1},a_s\?]\to\A\@$.

Unless otherwise stated, we shall only consider deformations leaving the interval $\@[\?t_0,t_1\@]\@$ fixed, i.e.~satisfying $\@a_0\/(\xi)=
t_0\,,\,a_N\/(\xi)=t_1\@$ $\forall\,\xi\@$. \,No restriction will be posed on the functions $\@a_s\/(\xi)\@,\,s=1,\ldots,N-1\@$.
Each curve $\@x_s\/(\xi):=\g_\xi\/(a_s\/(\xi))\,$ will be called the \emph{orbit\/} of the corner $\?x_s\/$ under the given deformation.

In local coordinates, setting $\@q^i\/(\g^{\ns}_{\;\xi}\/(t))=\vphi_{\ns}^{\;i}\/(\xi,t)\,$, the matching conditions (\ref{2.29}) read
\begin{equation}\label{2.30}
\vphi_{\ns}^{\;i}\?(\?\xi\?,\@a_s\/(\xi)\?)\,=\,\vphi_{\nS}^{\;i}\?(\?\xi\?,\@a_s\/(\xi)\?)\@,
\end{equation}
while the representation of the orbit $\@x_s\/(\xi)\@$ takes the form
\begin{equation}\label{2.31}
x_s\/(\xi)\@:\quad t\,=\,a_s\/(\xi)\,,\quad
q^i\@=\@\vphi_{\ns}^{\;i}\?(\?\xi\?,\@a_s\/(\xi)\?)\@.
\end{equation}

\medskip
The previous arguments have a natural ``infinitesimal'' counterpart. More specifically: an admissible infinitesimal deformation of an admissible closed arc
$\@\arc\g{\?a,b\@}\@$ is a triple $\@(\?\a,X,\b\@)\@$, where $\@X\@$ is the restriction to $\@[\?a,b\@]\@$ of an admissible infinitesimal deformation of
$\@\g:(c,d\@)\to\V\@$, while $\a,\?\b$ are the derivatives
\begin{equation}\label{2.32}
\a\,=\,\d\?a/d\xi\@\bigg|_{\@\xi=0}\,,\qquad\;\b\,=\,\d\?b/d\xi\@\bigg|_{\@\xi=0}\,,
\end{equation}
expressing the \emph{speed of variation\/} of the interval $\@\big[\@a\/(\xi)\?,\?b\/(\xi)\@\big]\@$ at $\@\xi=0\@$.

In a similar way, an admissible infinitesimal deformation of a piecewise differentiable evolution $\@\arc\g{\?t_0,t_1\?}\@$ is a collection \vspace{1pt}
$\@\big\{\!\?\cdots\@\a_{s-1}\@,\@X_{\ns}\@,\@\a_s\@\cdots\big\}\@$ of admissible infinitesimal deformations of each single closed arc, with \vspace{1.6pt}
$\@\a_s=\d\?a_s/d\xi\@\big|_{\@\xi=0}\,$ and, in particular, with $\@\a_0=\a_N=0\@$ whenever the interval $\@[\@t_0,t_1\?]\@$ is held fixed.

Let us analyse the situation in closer detail. To start with, we notice that the quantities $\@\a_s\?,\?X_{\ns}\@$ are not independent: Eqs.\;(\ref{2.30})
imply in fact the identities
\begin{equation*}
\de\@\vphi_{\ns}^{\;i}/de\xi\,+\,\de\@\vphi_{\ns}^{\;i}/de t\,\d\?a_s/d\xi\,=\,
\de\@\vphi_{\ns}^{\;i}/de\xi\,+\,\de\@\vphi_{\nS}^{\;i}/de t\,\d\?a_s/d\xi\@.
\end{equation*}

From these, evaluating everything at $\@\xi=0\@$ and recalling the relation between finite deformations and infinitesimal ones, we get the \emph{jump
relations\/}
\begin{equation}\label{2.33}
\Big(\?X_{\nS}^{\;i}\@-\@X_{\ns}^{\;i}\?\Big)_{a_s}\,=\,
-\@\a_s\@\bigg(\@\d\?q_{\nS}^{\;i}/dt\,-\,\d\?q_{\ns}^{\;i}/dt\@\bigg)_{\!a_s}
\,=\,-\@\a_s\@\Big[\?\psi^i\/(\h\g)\?\Big]_{\?a_s}\@.\hskip.6cm
\end{equation}

Furthermore, the admissibility of each single infinitesimal deformation $\@X_{\ns}\@$ requires the existence of a corresponding lift $\@\h
X_{\ns}\!=X_{\ns}^{\;i}\@\big(\de/de{q^i}\big)_{\h\g^{\ns}}+\@X_{\!\ns}^{\@A}\@\big(\de/de{z^A}\big)_{\h\g^{\ns}}\@$ satisfying the variational equation
(\ref{2.14}).

Both aspects are conveniently accounted for by assigning to each $\@\g^{\ns}\@$ an (arbitrarily chosen) infinitesimal control $\@h^{\ns}:V\/(\g^{\ns})\to
A\/(\h\g^{\ns})\,$\vspace{1pt}. In this way, proceeding as in Sec. \!2.5 and denoting by $\@\big(\D/D t\big)_{\g^{\ns}}\,$ the absolute time derivative along
$\@\g^{\ns}$ induced by $\@h^{\ns}\@$ we get the following
\begin{proposition}\label{Pro2.3}\hskip-.85pt
Every admissible infinitesimal deformation of an admissible evolution $\?\arc\g{\?t_0,t_1\?}\?$ over a fixed interval $\?[\?t_0,t_1\?]\?$ is determined, up to
initial data, by a collection of vertical vector fields $\@\big\{\@U_{\ns}\@ =\@U_{\ns}^{A}\,\big(\de/de{z^A}\big)_{\h\g^{\ns}}\big\}\/$, $\@s=1,\ldots,N\@$
and by $\@N-1\@$ real numbers $\@\a_1\?,\ldots,\?\a_{N-1}\@$ through the covariant variational equations
\begin{equation}\label{2.34}
\bigg(\D X_{\ns}/D t\bigg)_{\!\g^{\ns}}=\,\h\r\@\@(U_{\ns})\,=\,
U_{\ns}^{A}\@\bigg(\de\psi^i/de{z^A}\bigg)_{\!\h\g^{\ns}}
\bigg(\de/de{q^i}\bigg)_{\!\g^{\ns}}\qquad s=1,\ldots,N
\end{equation}
completed with the jump conditions (\ref{2.33}). The lift of the deformation is described by the family of vector fields
\begin{equation}\label{2.35}
\h X_{\ns}\,=\,h^{\ns}\?(X_{\ns})\@+\@U_{\ns}\@,\qquad\; s=1,\ldots,N\@.
\end{equation}
\end{proposition}
The proof is entirely straightforward, and is left to the reader. Introducing $n$ piecewise differentiable vector fields $\DE_1\?,\ldots,\?\DE_n\@$ along
$\?\h\g\?$ according to the prescription
\begin{equation*}
\DE_i\/(t)\@=\@
\@h^{\ns}\@\bigg(\de/de{q^i}\bigg)_{\!\g^{\ns}\/(t)}\qquad\quad \forall\;t\in(\?a_{s-1},a_s)\,,\;\;s=1,\ldots,N\@,
\end{equation*}
Eq.\;(\ref{2.35}) takes the explicit form
\begin{equation}\label{2.36}
\h X_{\ns}\,=\,h^{\ns}\bigg(X_{\ns}^{\;i}\@\bigg(\de/de{q^i}\bigg)_{\!\g}\bigg)\@+\@U_{\ns}
\,=\,X_{\ns}^{\;\@i}\,\DE_i\,+\,U_{\ns}^A\@\bigg(\De/de{z^A}\bigg)_{\!\h\g}
\end{equation}
on each open arc $\@\h\g^{\ns}:(a_{s-1},a_s)\to\A\@$.

\smallskip
To discuss the implications of Eq.\;(\ref{2.34}) we resume the notation $\@\Vg\@$ and $\@V^*\/(\g)\@$ respectively for the totality of vertical vectors and for
the totality of virtual $1$--forms along $\g\?$
\footnote%
{\@Notice that this definition is perfectly meaningful even at the corners $\@\g\/(a_s)\@$.}.
We then define a transport law for virtual tensors on $\@\g\@$, henceforth called $h$--transport, gluing $\@h^{\ns}$--transport along each arc
$\@\arc{\g^{\ns}}{\?a_{s-1}\?,\?a_s\@}\@$ and continuity at the corners, i.e.~continuity of the components at $\?t=a_s\@$.

In view of Proposition \ref{Pro2.2}, the $h$--transported vector fields form an $\?n$--dimensional vector space $\@V_h\@$, isomorphic to each fibre
$\@\Vg_{|t}\,$. Likewise, the $h$--transported virtual $1$--forms span a vector space $\@V^*_h\@$, dual to $\@V_h\@$ and isomorphic to
$\@V^*\/(\g)_{|t}\;\@\forall\,t\@$.

This provides canonical identifications of $\@\Vg\@$ with the cartesian product $\@[\?t_0,t_1\@]\times V_h\@$ and of $\@V^*\/(\g)\@$ with the product
$\@[\?t_0,t_1\@]\times V^*_h\@$, thus allowing to regard every section $\@X:[\?t_0,t_1\@]\to\Vg\,$ as a vector valued map $\@X:[\?t_0,t_1\@]\to V_h\@$, and every
section $\@\h\l:[\?t_0,t_1\@]\to V^*\/(\g)\@$ as a map $\@\h\l:[\?t_0,t_1\@]\to V^*_h\@$.

Exactly as in Sec. \!2.3, the situation is formalized referring $\@V_h\@$ and $\@V^*_h\@$ to dual bases $\@\{\@e_{(a)}\@\}\@$, $\@\{\@e^{(a)}\@\}\@$ related to
$\@\big(\de/de{q^i}\big)_{\g}\@$, $\wg i\@$ by the transformation
\begin{equation}\label{2.37}
\bigg(\de/de{q^i}\bigg)_{\!\g}\@=\@\e^a_i\?(t)\;e_{(a)}\;,\qquad\;\wg i\@=\@\e_a^i\?(t)\;e^{(a)}\@.
\end{equation}

Given any admissible infinitesimal deformation $\@\big\{(X_{\ns},\a_s)\big\}\/$, we now glue all sections 
\linebreak
$\@X_{\ns}:[\@a_{s-1}\?,\?a_s\@]\to V\?(\g^{\ns})\@$
into a single, piecewise differentiable map $\@X:[\@t_0,t_1\@]\to V_h\@$, with jump discontinuities at $\@t=a_s\@$ satisfying Eq.\;(\ref{2.33}).

For each $\@s=1,\ldots,N\@$ we have then the relations
\begin{equation}\label{2.38}
X_{\ns}=X^a\/(t)\,e_{(a)}\@,\quad\bigg(\D X_{\ns}/D t\bigg)_{\!\g^{\ns}}\!=\d X^a/d t\,e_{(a)}\qquad\;\forall\;t\in (\?a_{s-1}\?,\?a_s\@)\@.
\end{equation}

In a similar way, we collect all fields $U_{\!\@\ns}$ into a single object $\?U$, conventionally called a vertical vector field along $\@\h\g\@$.

In this way, the covariant variational equation (\ref{2.34}) takes the form
\begin{subequations}\label{2.39}
\begin{equation}
\d X^a/d t\,=\,U^A\,\e^a_i\,\bigg(\de\psi^i/de{z^A}\bigg)_{\!\h\g}
\qquad\;\forall\;t\ne a_s\@,
\end{equation}
completed with the jump conditions
\begin{equation}
\big[\?X^a\@\big]_{\?a_s}=\big[\?X^i\@\big]_{\?a_s}\,\e^a_i\?(a_s)=-\@\a_s\,\e^a_i\?(a_s)\,\Big[\?\psi^i\/(\h\g)\?\Big]_{a_s}\;\,\;s=1,\ldots,N-1\@.
\end{equation}
\end{subequations}

\smallskip
\subsection{The variational setup}\label{Sec2.7}
\textbf{(\/i\/)} \,Given an admissible piecewise differentiable section 
\linebreak
$\g:[\@t_0\@,\@t_1\@]\to\V\@$, let $\@\VV\@$ denote the infinite--dimensional vector
space formed by the totality of vertical vector fields $\@U=\big\{U_{\ns}\@,\,s=1,\ldots,N\@\big\}\@$ along $\@\h\g\@$. Also, let $\@\W\@$ denote the direct
sum $\@\VV\oplus\R^{N-1}$.

By Eqs.\;(\Ref{2.39}a,b), every admissible infinitesimal deformation $X$ of $\@\g\@$ is then determined, up to initial data, by an element
$\@(U,\@\atilde\?):=(U,\a_1,\ldots,\a_{N-1})\in\W\@$.

For variational purposes, let us now focus on the class of infinitesimal deformations $\@X:[\?t_0\?,\?t_1\@]\to\Vg\@$ {\em vanishing at the endpoints of}
$\g\@$. Setting $\@X\/(t_0)=0\@$, Eqs.\;(\Ref{2.39}a,b) provide the evaluation
\begin{equation}\label{2.40}
X\/(t)\@=\@\left(\int_{t_0}^t\,U^A\,\e^a_i\,\bigg(\de\psi^i/de{z^A}\bigg)_{\!\h\g}\,d\/t\,
-\,\sum_{a_s<\?t}\,\a_s\;\e^a_i\?(a_s)\,\Big[\?\psi^i\/(\h\g)\?\Big]_{a_s}\right)e_{(a)}\@.
\end{equation}
Denoting by $\@\U:\W\to V_h\@$ linear map defined by the equation
\begin{equation}\label{2.41}
\U\/(U,\atilde\?)=\left(\;\int_{t_0}^{t_1}U^A\,\e^a_i\@\bigg(\de\psi^i/de{z^A}\bigg)_{\!\h\g}\@d\/t\;-\@
\sum_{s=1}^{N-1}\@\a_s\,\@\e^a_i\?(a_s)\@\Big[\?\psi^i\/(\h\g)\?\Big]_{a_s}\,\right)e_{(a)}\@,
\end{equation}
it is then an easy matter to conclude that the required class is in 1--1 correspondence with the subspace $\@\ker\@(\U)\subset\W\@$.

Actually, as anticipated in Remark \ref{Rem2.1}, what really matters in a variational context is not the space $\@\ker\@(\U)\@$ itself, but the (possibly
smaller) subset $\@\X\subset\ker\@(\U)\@$ formed by the infinitesimal deformations tangent to admissible \emph{finite\/} deformations with fixed endpoints.

The linear span of $\@\X\@$, henceforth denoted by $\@\Delta\/(\g)\@$, will be called the \emph{variational space\/} of $\@\g\@$. The evolutions of the system
will be classified into \emph{ordinary\/}, when $\@\Delta\/(\g) =\ker\@(\U)\@$ and~\emph{exceptional\/}, when $\@\Delta\/(\g)\subsetneq\ker\@(\U)\@$.

\medskip\noindent
\textbf{(\/ii\/)} \,A further important classification comes from the nature of the inclusion $\U\?(\W)\subset V_h\@$: an evolution
$\@\arc{\?\g}{\?t_0\?,\?t_1}\@$ is called \emph{normal\/} if $\U\?(\W)=V_h$,
\emph{abnormal\/} in the opposite case\@
\footnote%
{\@As we shall see, when applied to the extremals of an action functional, the terminology agrees with the current one (see, among others, \cite{Montgomery,SL}
and references therein).}.
It is called \emph{locally normal} if its restriction to any closed subinterval $\@[\?t_0'\?,\?t_1'\?]\subseteq [\?t_0\?,\?t_1\?]\@$ is normal.
The dimension of the annihilator $\@\big(\@\U\?(\W)\@\big){}^0\subset V_h^*\@$ is called the \emph{abnormality index\/} of $\@\g\@$.

In this connection, a useful result is provided by the following
\begin{proposition}\label{Pro2.4}
The annihilator $\@\big(\@\U\?(\W)\@\big){}^0\subset V_h^*\@$ coincides with the totality of $\?h$--transported virtual $1$--forms $\@\h\rho=\rho_{\?i}\,\wg
i\@$ satisfying the conditions
\begin{subequations}\label{2.42}
\begin{align}
&\rho_{\?i}\@\bigg(\de\psi^i/de{z^A}\bigg)_{\!\h\g}=\,0\@,\qquad&&A=1,\ldots,r\@,       \\[4pt]
&\rho_{\?i}\?(a_s)\?\big[\@\psi^i\/(\h\g)\@\big]_{\?a_s}=\,0\@,\qquad&& s=1,\ldots,N-1\@.
\end{align}
\end{subequations}
\end{proposition}
\begin{proof}
In view of Eq.\;(\Ref{2.41}), the subspace $\@\big(\@\U\?(\W)\@\big){}^0\subset V_h^*\@$ consists of the totality of elements
$\@\h\rho=\rho_{\?a}\,e^{(a)}=\rho_{\?a}\,\e^a_i\?\wg i\@$ satisfying the relation
\begin{equation*}
\rho_{\?a}\left(\,\int_{t_0}^{t_1}U^A\,\e^a_i\bigg(\de\psi^i/de{z^A}\bigg)_{\!\h\g}\@d\/t
\;-\@\sum_{s=1}^{N-1}\@\a\?_s\,\@\e^a_i\@(a\?_s)\@\Big[\?\psi^i\/(\h\g)\?\Big]_{a_s}\@\right)=\@0\qquad \forall\;(U,\atilde)\in\W\@,
\end{equation*}
clearly equivalent to Eqs.\;(\Ref{2.42}a,b).
\end{proof}

On account of Eqs.\;(\Ref{2.24}b), (\Ref{2.25}b), the condition of $h$--transport of $\@\h\rho\@$ along each arc $\@\g^{\ns}\@$ is expressed in coordinates as
\begin{equation}\label{2.43}
\d\rho_{\?i}/d t\,+\,\rho_k\left(\de\psi^k/de{q^i}\right)_{\!\h\g}\,+\,
h_i{}^A\cancel{\,\rho_k\left(\de\psi^k/de{z^A}\right)_{\!\h\g}}\,=\,0\@,
\end{equation}
the cancellation being due to Eq.\;(\Ref{2.42}a). The content of Proposition \ref{Pro2.4} is therefore independent of the choice of the infinitesimal controls
$\@h^{\ns}:V\/(\g^{\ns})\to A\@(\h\g^{\ns})\,$.

\smallskip
This aspect is formalized denoting by $\@\dot\g=\big(\de/de t\big)_\g +\psi^{\@i}_{|\h\g}\,\big(\de /de{q^i}\big)_{\g}\@$\vspace{1.5pt} the tangent vector to
the curve $\@\g\@$. Recalling the expression (\ref{2.10}) for the symbolic time derivative of a $1$--form along $\@\g\@$, we have then the following
\begin{corollary}\label{Cor2.1}
The annihilator $\@\big(\@\U\?(\W)\@\big){}^0\@$ is isomorphic to the vector space formed by the totality of continuous (ordinary) $1$--forms
$\@\rho=\rho\@_0\@\/(t)d\/t_{|\g}+\rho_i\/(t)\@d\/q^{\@i}_{|\g}\@$ along $\@\g\@$ satisfying the conditions
\begin{equation}\label{2.44}
\big<\?\rho\@,\@\dot\g\?\big>\@=\@0\,,\qquad\; \d\rho/dt\,=\,0\@.
\end{equation}
\end{corollary}
\begin{proof}
In coordinates, Eqs.\;(\ref{2.44}) are summarized into the system
\begin{align*}
\rho_0\@&=\@-\@\rho_i\,\psi^{\@i}_{|\h\g}\,,                                                                                   \\[4pt]
0\@&=\@-\d/dt\,\big(\rho_i\@\psi^{\@i}_{|\h\g}\big)\@d\/t_{|\h\g}\@+\@\d\rho_i/dt\,d\/q^{\@i}_{|\h\g}\@+\@\rho_i\,d\?\psi^{\@i}_{|\h\g}\,= \\[4pt]
&=\@\biggl[\d\rho_i/dt\@+\@\rho_k\biggl(\de\psi^k/de{q^i}\biggr)_{\!|\h\g}\biggr]\Big(d\/q^i-\@\psi^i\@d\/t\Big)_{|\h\g}\@+
\@\rho_k\biggl(\de\psi^k/de{z^A}\biggr)_{\!|\h\g}\biggl(d\/z^A-\@\d z^A/dt\,d\/t\biggr)_{\!|\h\g}\@.
\end{align*}
From the latter it is easily seen that the quotient map $\@\varpi:T^*_\g\/(\V)\to V^*\/(\g)\@$ sets up a 1-1 correspondence between continuous $1$--form
satisfying Eqs.\;(\ref{2.44}) and virtual $1$--forms $\@\h\rho:=\rho_i\/(t)\@\wg i\@$ satisfying Eqs.\;(\ref{2.42}a,b), (\ref{2.43}).
\end{proof}

\vskip-4pt
Straightforward consequences of Proposition \ref{Pro2.4} are:
\begin{itemize}
\item
in the absence of constraints, all evolutions are automatically \emph{normal\/};
\item
the abnormality index of a piecewise differentiable section $\@\g\@$ cannot exceed the abnormality index of each single arc $\@\g^{\ns}\@$.
\end{itemize}

\noindent
These and other topics related to normality are discussed in Appendix B.\vspace{2pt}

An important insight into the hierarchy between normality and ordinariness is provided by the following
\begin{proposition}\label{Pro2.5}
The normal evolutions form a subset of the ordinary ones.
\end{proposition}
The result is proved in Appendix A. In this connection, see also \cite{Hsu}.

\medskip

\section{Calculus of variations}\label{Sec3}
\subsection{Extremaloids}\label{Sec3.1}
Let us now come to the central problem of variational calculus. Let $\@\Lagr\in\F\/(\A)\,$ denote a differentiable function on the velocity space $\A\@$,
henceforth called the \emph{Lagrangian\/}.
Also, let $\?\arc{\@\g}{\?t_0\?,\?t_1}\@$ ($\@\g\?$ for short) denote an admissible piecewise differentiable evolution of the system, defined on a closed
interval $\@[t_0,t_1]\subset\R\@$.

Indicating by $\@\h\g\@$ the lift of $\@\g\@$, consider the \emph{action functional\/}
\begin{equation}\label{3.1}
\I\@[\g]:= \int_{\h\g} \Lagr\@dt
\,:=\,\sum_{s=1}^N\;\int_{a_{s-1}}^{\@a_s}\@\big(\h\g^{\ns}\big)^*\/(\Lagr)\,dt
\end{equation}
\begin{definition}\label{Def3.1}
The evolution $\?\g\?$ is called an \emph{extremaloid} for the functional (\ref{3.1})\vspace{1.8pt} if and only if, for every admissible deformation $\?\gxi=\?$
$\big\{\arc{\?\g^{\ns}_{\;\xi}}{\@a_{s-1}\/(\xi),a_s\/(\xi)\@}\?\big\}\?$ \vspace{1pt} with fixed endpoints, the function
\begin{equation*}
\I\@(\xi)\@:=\@\int_{\hgxi} \Lagr\@dt =\sum_{s=1}^N\;\int_{a_{s-1}\/(\xi)}^{\@a_s\/(\xi)}\@
\big(\h\g_{\;\xi}^{\ns}\big)^*\/(\Lagr)\,dt
\end{equation*}
is stationary at $\xi=0\@$.
\end{definition}

The content of Definition \ref{Def3.1} is formalized denoting by $\h X_{\ns}\@$ the infinitesimal deformation associated with each $\?\h\g^{\ns}_{\;\xi}\@$.
Recalling Eq.\;(\ref{2.36}) as well as the definition $\/\a_s=\d a_s/d\xi\big|_{\xi=0}\,$, we have then the evaluation
\begin{subequations}\label{3.2}
\begin{multline}
\d \/\I/d\xi\bigg|_{\xi=0}\,=\,\sum_{s=1}^N\left[\, \d/d\xi\int_{a_{s-1}\/(\xi)}^{\@a_s\/(\xi)}\@
\Lagr\/\big(\h\g^{\ns}_{\;\xi}\big)\,d\/t\,\right]_{\@\xi=0}\,=                         \\[3pt]
=\,\sum_{s=1}^N\@\bigg\{\@\int_{a_{s-1}}^{\@a_s}\@\h X_{\ns}\/(\Lagr)\,d\/t\;+\;\Big[\@
\a_s\,\Lagr\/(\h\g^{\ns}\/(a_s))-\a_{s-1}\,\Lagr\/(\h\g^{\ns}\/(a_{s-1}))\@\Big]\@\bigg\}\@.
\end{multline}
On account of the requirement $\@\a_0=\a_N=0\@$, denoting by
\begin{equation*}
\big[\?\Lagr\/(\h\g)\?\big]_{\?a_s}:=\big[\@\Lagr\/(\h\g^{\nS}\/(a_s))-
\Lagr\/(\h\g^{\ns}\/(a_s))\@\big]
\end{equation*}
the jump of the function $\@\Lagr\/(\h\g\/(t))\@$ at $t=a_s$, Eq.\;(\Ref{3.2}a) may be concisely written as
\begin{equation}
\d \/\I/d\xi\bigg|_{\xi=0}\,=\,
\sum_{s=1}^N\@\int_{a_{s-1}}^{\@a_s}\@\Big(\@X_{\ns}^{\;i}\@\DE_i\/(\Lagr)+U_{\ns}^A\@\de
\Lagr/de{z^A}\@\Big)\,d\/t\,-\,\sum_{s=1}^{N-1}
\a_s\@\big[\?\Lagr\/(\h\g)\?\big]_{\?a_s}\@.
\end{equation}
\end{subequations}

To handle Eq.\;(\Ref{3.2}b),  we introduce $N$ virtual $1$--forms $\@\h\rho\?^{\ns}\!=p_{\;i}^{\ns}\!(t)\@\wgs i\@$ (one for each arc $\@\g^{\ns}\@$)
satisfying the transport law\,%
\footnote%
{\,For the notation, see Eq.\;(\ref{2.21})}
\begin{subequations}\label{3.3}
\begin{align}
&\bigg(\D\@\h\rho\@^{\ns}/Dt\bigg)_{\!\g^{\ns}}\,=\,h^*(d\?\Lagr)_{|\?\h\g^{\ns}}\,=
\,\Big(\?\DE_i\@\Lagr\@\Big)_{\h\g^{\ns}}\,\wgs i \\
\intertext{as well as the matching conditions}
&\h\rho\@^{\ns}\big|_{\?a_s}\,=\,\h\rho\@^{\nS}\big|_{\?a_s}\@,\qquad\;s=1\?,\ldots,\?N-1\@.
\end{align}
\end{subequations}

Once again, for notational convenience, we collect all $\@\h\rho\@^{\ns}\@$ into a continuous, piecewise differentiable section $\@\h\rho:[\@t_0,t_1\@]\to
V^*\/(\g)\@$ according to the prescription
\begin{equation}\label{3.4}
\h\rho\?(t)\@=\@\h\rho\@^{\ns}\/(t)
\qquad\quad\forall\;t\in[\@a_{s-1}\?,\@a_s\@]\@.
\end{equation}

On account of Eqs.\;(\ref{3.3}\,a,b), $\@\h\rho\@$ is then uniquely determined by $\@\Lagr\@$, up to an arbitrary $h$--transported virtual $1$--form along
$\@\g\@$.

Taking the covariant variational equation (\ref{2.34}) as well as the duality relations 
\linebreak
$\@\Big<\Big(\de/de{q^i}\Big)_{\!\g^{\ns}}\@,\@\wgs k\Big>
=\delta\?^k_i\@$ into account, by Eq.\;(\Ref{3.3}a) we get the expression
\begin{align*}
X_{\ns}^{\;i}\,\DE_i\@\Lagr=
\bigg<\!X_{\ns}\@,\bigg(\!\D\h\rho\@^{\ns}/Dt\bigg)_{\!\!\g^{\ns}}\bigg>
&=\d/d t\,\Big<\@X_{\ns}\@,\@\h\rho\@^{\ns}\@\Big>-
\bigg<\!\bigg(\!\D X_{\ns}/Dt\!\bigg)_{\!\!\g^{\ns}}\!,\@\h\rho\?^{\ns}\bigg>=       \\[4pt]
&=\,\d/d t\,\Big(X_{\ns}^{\;i}\;p_{\;i}^{\ns}\@\Big)\,-\,
p_{\;i}^{\ns}\bigg(\de\psi^i/de{z^A}\bigg)_{\!\h\g^{\ns}}\@U_{\ns}^{A}\quad
\end{align*}

\vskip-10pt\noindent whence also\vspace{4pt}
\begin{equation*}
\int_{a_{s-1}}^{\@a_s}X_{\ns}^{\;i}\;\DE_i\/(\Lagr)\;d\/t\;=\,
\bigg[\@X_{\ns}^{\;i}\;p_{\;i}^{\ns}\@\bigg]_{a_{s-1}}^{a_s}-\;
\int_{a_{s-1}}^{\@a_s}\@p_{\;i}^{\ns}\bigg(\de\psi^i/de{z^A}\BIGG)_{\h\g^{\ns}}
\@U_{\ns}^{A}\;d\/t\@.
\end{equation*}

Summing over $s$ and recalling Eqs.\;(\ref{2.33}), (\ref{3.3}b) as well as the conditions $\@X\/(t_0)=X\/(t_1)=0\@$, this implies the relation
\begin{equation*}
\sum_{s=1}^N\,\int_{a_{s-1}}^{\@a_s}X_{\ns}^{\;i}\;\DE_i\/(\Lagr)\;d\/t\,=\,
-\int_{t_0}^{t_1}p_i\@\bigg(\de\psi^i/de{z^A}\bigg)_{\!\h\g}\@U^A\,d\/t\,+\,
\sum_{s=1}^{N-1}\a_s\@\Big[\?\psi^i\/(\h\g)\?\Big]_{a_s}\@p_i\/(a_s)\@.
\end{equation*}

In this way, omitting all unnecessary subscripts, Eq.\;(\Ref{3.2}b) attains the final form
\begin{equation}\label{3.5}
\d \/\I/d\xi\bigg|_{\xi=0}=\int_{t_0}^{t_1}\!\bigg(\de \Lagr/de{z^A}\@-\@p_i\,
\de\psi^i/de{z^A}\bigg)\@U^A\,d\/t\,+\,\sum_{s=1}^{N-1}
\a_s\@\Big[\,p_i\/(t)\,\psi^i\/(\h\g)\@-\@\Lagr\/(\h\g)\@\Big]_{\?a_s}\@.
\end{equation}

In the algebraic environment introduced in Sec. \!2.7, the previous discussion is naturally formalized regarding the right--hand side of Eq.\;(\ref{3.5}) as a
linear functional $\@d\@\I_\g:\W\to\R\@$ on the vector space $\@\W=\VV\oplus\R^{N-1}$.

A necessary and sufficient condition for $\@\g\@$ to be an extremaloid for the functional (\ref{3.1}) is then the vanishing of $\@d\@\I_\g\@$ on the subset
$\@\X\subset\W\@$ formed by the totality of elements$\@\@(U,\@\atilde\?)\@$ tangent to admissible finite deformations with fixed endpoints. By linearity, this
condition is mathematically equivalent to the requirement
\begin{equation}\label{3.6}
\Delta\/(\g)\@\subset\@\ker\@(d\@\I_\g)\@,
\end{equation}
$\@\Delta\/(\g)=\text{Span}\@(\X)\subseteq\ker\@(\U)\@$ denoting the variational space of $\@\g\@$.

As we shall see, Eq.\;(\ref{3.6}) provides an algorithm for the determination of \emph{all\/} the extremaloids of the functional (\ref{3.1}) within the class
of \emph{ordinary\/} evolutions.

The exceptional case is more complicated: the lack of an explicit characterization of the space $\@\Delta\/(\g)\@$ in terms of local properties of the
section~$\g\@$ requires in fact a direct check of Eq.\;(\ref{3.6}) on each exceptional evolution.

In what follows we pursue an intermediate strategy; namely, rather than dealing with Eq.\;(\ref{3.6}) we discuss the implications of the stronger requirement
\begin{subequations}\label{3.7}
\begin{equation}
\ker\@(\U)\@\subset\@\ker\@(d\@\I_\g)\@.
\end{equation}

According to the classification introduced in Sec. \!\ref{Sec2.7}, the latter is \emph{necessary and sufficient\/} for an \emph{ordinary\/} evolution, but merely
\emph{sufficient\/} for an \emph{exceptional\/} one, to be an extremaloid of the functional (\ref{3.1}).

By elementary algebra, the requirement (\Ref{3.7}a) is equivalent to the existence of a (possibly non--unique) linear functional $\@\h\s:V_h\to\R\@$ ---
i.e.\;of a $h$--transported virtual $1$--form along $\@\g\@$ --- satisfying the relation
\begin{equation}
\begin{picture}(40,11)
\put (11,7.5){\vector(1,0){11}}     \put (15,9.5){$\U$}
\put (26,5){\vector(0,-1){8}}       \put (28,0){$\h\s$}
\put (8,5){\vector(3,-2){15}}       \put (9,-2.7){$d\@\I_\g$}
\put (7,7.5){\makebox(0,0){\hss$\W$\hss}}
\put (26,7.5){\makebox(0,0){\hss$V_h$\hss}}
\put (26,-5.6){\makebox(0,0){\hss$\R$\hss}}
\end{picture}\plus{36}{36}
\end{equation}
\end{subequations}

Setting $\@\h\s=\s_a\,e^{(a)}=\s_i\@\wg i\@$ and recalling Eqs.\;(\ref{2.41}), (\ref{3.5}), Eq.\;(\Ref{3.7}b) amounts to the condition
\begin{multline*}
\int_{t_0}^{t_1}\bigg(\@\de \Lagr/de{z^A}\@-\@p_i\,\de\psi^i/de{z^A} \@\bigg)U^A\,d\/t\,+\,\sum_{s=1}^{N-1}
\a_s\@\Big[\@p_i\/(t)\@\psi^i\/(\h\g)\@-\@\Lagr\/(\h\g)\@\Big]_{\?a_s}\,=                                   \\[2pt]
=\,\s_a\left(\;\int_{t_0}^{t_1}U^A\,\e^a_i\@\bigg(\de\psi^i/de{z^A}\BIGG)_{\h\g}\@d\/t\;-\@
\sum_{s=1}^{N-1}\@\a_s\,\@\e^a_i\?(a_s)\@\Big[\?\psi^i\/(\h\g)\?\Big]_{a_s}\,\right)\,=                     \\[2pt]
=\,\int_{t_0}^{t_1}U^A\,\s_i\@\bigg(\de\psi^i/de{z^A}\BIGG)_{\h\g}\@d\/t\;-\@
\sum_{s=1}^{N-1}\@\a_s\,\@\s_i\?(a_s)\@\Big[\?\psi^i\/(\h\g)\?\Big]_{a_s}\@.\quad
\end{multline*}
By the arbitrariness of $\@(U,\@\atilde\?)\@$, the latter splits into the system

\vskip-14pt
\begin{subequations}\label{3.8}
\begin{align}
&\de \Lagr/de{z^A}\@-\big(\@p_i\@+\s_i\big)\@\de\psi^i/de{z^A}\,=\,0\,,
\qquad&&A=1,\ldots,r\,,                                                                                     \\[5pt]
&\Big[\@\big(\@p_i\@+\s_i\big)\@\psi^i\/(\h\g)\@
-\@\Lagr\/(\h\g)\@\Big]_{\?a_s}\,=\,0\,,\qquad&&s=1,\ldots,N-1\,.\qquad\quad
\end{align}
\end{subequations}

\smallskip\noindent
Collecting all results, and recalling Propositions \ref{Pro2.4}, \ref{Pro2.5} we conclude
\begin{theorem}\label{Teo3.1}
Given an admissible evolution $\@\g\@$, let $\@\p\?$ denote the totality of piecewise differentiable virtual $\?1$--forms $\?\h\rho=p_i\/(t)\,\wg i$ along
$\?\g$ satisfying Eqs.\;(\ref{3.3}a,b), (\ref{3.4}) as well as the finite relations
\begin{subequations}\label{3.9}
\begin{align}
&p_i\,\de\psi^i/de{z^A}\,=\,\de \Lagr/de{z^A}\,,\qquad&&A=1,\ldots,r                                        \\
\intertext{and the matching conditions}
&\Big[\@p_i\@\psi^i\/(\h\g)\@-\@\Lagr\/(\h\g)\@\Big]_{\?a_s}\,=\,0\,,\qquad&&s=1,\ldots,N-1\,.
\end{align}
\end{subequations}
Then: \ALPH
\begin{enumerate}
\item the condition $\?\p\ne\emptyset\@$ is \emph{sufficient\/} for $\?\g\?$ to be an extremaloid for the functional (\ref{3.1})\@;
\item if $\@\g\@$ is an \emph{ordinary\/} evolution, the condition $\?\p\ne\emptyset\@$ is also \emph{necessary\/} for $\@\g\@$ to be an extremaloid;
\item $\g\?$ is a \emph{normal extremaloid\/}, namely an extremaloid belonging to the class of normal evolutions, if and only if the set $\@\p\@$ consists of a
    single element.
\end{enumerate}
\end{theorem}
\begin{proof}
In view of Eqs.\;(\ref{3.5}), (\Ref{3.9}a,b), the existence of at least one $\@\h\rho\in\p\@$ ensures the validity of $\@\d \I/d\xi\big|_{\?\xi=0}\@
=\@0\plus{10.8}{6.8}\,$ for all admissible infinitesimal deformations vanishing at the endpoints of $\?\g\?$. Assertion a) is then a direct consequence of
Definition~\ref{Def3.1}.

In particular, according to our previous discussion, if $\@\g\@$ is an ordinary extremaloid, for any continuous virtual $1$-form  $\@\h\rho= p_i\@\?\wg i\@$
obeying the transport law (\Ref{3.3}a) there exists at least one $h$--transported $1$-form  $\@\h\s=\s_i\@\wg i\@$ satisfying Eqs.\;(\Ref{3.8}a,b). The sum
$\@\h\rho+\h\s=\big(\@p_i\@+\s_i\@\big)\,\wg i\@$ is then automatically in the class $\p$, thus proving assertion b).

Finally, as pointed out in Sec. \!2.7, the normal evolutions form a subclass of the ordinary ones, uniquely characterized by the requirement
$\@\big(\?\U\?(\W)\?\big){}^0=\{\?0\?\}\@$.

Therefore, according to assertion b), a normal evolution $\@\g\@$ is an extremaloid if and only if the class $\p$ is nonempty.
Moreover, by Eqs.\;(\Ref{3.3}a), (\Ref{3.8}a), given any pair $\@\h\rho\@,\@\h\rho'\in\p\@$, the difference $\@\h\rho-\h\rho'\@$ is an $h$--transported
$1$--form satisfying Eqs.\;(\Ref{2.42}a,b).
By Proposition \ref{Pro2.4} this implies 
\linebreak
$\@\h\rho-\h\rho'\@\in\@\big(\@\U\?(\W)\@\big){}^0\@\Rightarrow\@\h\rho=\h\rho'\@$, thus establishing assertion c).
\end{proof}

In view of Eqs.\;(\Ref{2.24}b), (\Ref{2.25}b), for any $\@\h\rho\in\p\@$ the transport law (\Ref{3.3}a) simplifies to
\begin{equation*}
\d p_i/d t\,+\,p_k\left(\de\psi^k/de{q^i}\right)_{\!\h\g}\,+\,h_i{}^A\cancel{\,p_k\left(\de\psi^k/de{z^A}\right)_{\!\h\g}}\,=\,
\left(\de \Lagr/de{q^i}\right)_{\!\h\g}\,+\,\cancel{h_i{}^A\@\left(\de \Lagr/de{z^A}\right)_{\!\h\g}}\@,
\end{equation*}
the cancellation being due to Eq.\;(\Ref{3.9}a). Exactly as it happened with Proposition~\ref{Pro2.4}, all assertions of Theorem \ref{Teo3.1} have therefore an
intrinsic meaning, irrespective of the choice of the infinitesimal controls $\@h^{\ns}:V\/(\g^{\ns})\to A\/(\h\g^{\ns})\,$.

The previous arguments show that the determination of the \emph{ordinary\/} extremaloids of the functional (\ref{3.1}) relies on $2\@n+r\/$ equations

\vskip-12pt
\begin{subequations}\label{3.10}
\begin{align}
& \d q^i/d t\,=\,\psi^i\/(t,q^i,z^A)\,,                                             \\[1pt]
& \d p_i/d t\,+\,\de\psi^k/de{q^i}\,p_k\,=\,\de \Lagr/de{q^i}\,,                    \\[1pt]
&p_i\,\de\psi^i/de{z^A}\,=\,\de \Lagr/de{z^A}
\end{align}
\end{subequations}
for the unknowns $\@q^i\/(t)\@,p_i\/(t)\@,\@z^A\/(t)\@$, completed with the continuity requirements
\begin{equation}\label{3.11}
\big[\@q^i\@\big]_{\?a_s}\,=\,\big[\@p_i\@\big]_{\?a_s}\,=\,\big[\@p_i\,\psi^i-\Lagr\@\big]_{\?a_s}\,=\,0\@,\qquad\quad s=1,\ldots,N-1\@.
\end{equation}

As already pointed out, all equations are independent of the choice of the infinitesimal controls, and involve only the ``true'' data of the problem, namely
the Lagrangian $\?\Lagr\@$ and the constraint equations (\ref{2.2}). In particular:
\begin{itemize}
\item
the algorithm (\Ref{3.10}a,b,c), (\ref{3.11}) is invariant under arbitrary transformations of the form
\begin{equation}\label{3.12}
\Lagr\,\to\,\Lagr\@+\,\de f/de t\,+\,\de f/de{q^k}\,\psi^k\,,\qquad p_i\/(t)\,\to\,p_i\/(t)\@+\@\bigg(\de f/de{q^i}\bigg)_{\!\g\/(t)}\@,
\end{equation}
$f\/(t,q^1,\ldots,q^n)\@$ being any differentiable function over $\@\V\@$;\vspace{5pt}

\item
the last pair of equations (\ref{3.11}) extend to the present context the well known \emph{Erdmann--Weierstrass corner conditions\/} of holonomic variational
calculus \cite{Sagan,Giaquinta}.
\end{itemize}

In strict analogy with the result established in Corollary \ref{Cor2.1}, a more compact description of the extremality conditions is provided by the following
\begin{corollary}\label{Cor3.1}
The class $\@\p\@$ is in 1-1 correspondence with the (possibly empty) affine space $\@\wp^*\/(\g)\@$ formed by the totality of \emph{continuous\/} $1$--forms
$\@\rho=p_0\@d\/t_{|\g}+p_i\@d\/q^{\@i}_{|\g}\@$ along $\@\g\@$ satisfying the conditions
\begin{equation}\label{3.13}
\big<\?\rho\@,\@\dot\g\?\big>\@=\@\Lagr_{|\h\g}\,,\qquad\; \d\rho/dt\,=\,d\/\Lagr_{|\h\g}\@.
\end{equation}
\end{corollary}
\begin{proof}
In coordinates, Eqs.\;(\ref{3.13}) are summarized into the system
\begin{subequations}\label{3.14}
\begin{align}
p_0\@&=\@\Lagr_{|\h\g}-\@p_i\,\psi^{\@i}_{|\h\g}\,,                                                                                                \\[4pt]
0\;&=\@d\/\Lagr_{|\h\g}-
\d/dt\,\big(\Lagr_{|\h\g}-p_i\@\psi^{\@i}_{|\h\g}\big)\@d\/t_{|\h\g}\@-\@\d p_i/dt\,d\/q^{\@i}_{|\h\g}\@-\@p_i\,d\?\psi^{\@i}_{|\h\g}\,=        \nn\\[4pt]
&\quad\;=\biggl[\biggl(\de \Lagr/de{q^i}\biggr){\plus07}_{\!|\h\g}-\d p_i/dt-p_k\biggl(\de\psi^k/de{q^i}\biggr){\plus07}_{\!|\h\g}\biggr]
\Big(d\/q^i\!\@-\psi^i\@d\/t\Big){\!\plus05}_{|\h\g}\@+                                                                                         \nn\\[4pt]
&\quad\;+\@\biggl[\biggl(\de \Lagr/de{z^A}\biggr){\plus07}_{\!|\h\g}-
p_k\biggl(\de\psi^k/de{z^A}\biggr){\plus07}_{\!|\h\g}\biggr]\biggl(d\/z^A-\@\d z^A/dt\,d\/t\biggr){\plus07}_{\!|\h\g}\@.
\end{align}
\end{subequations}

Eq.\;(\ref{3.14}a,b), together with the assumed continuity of the components $\@p_0\@,p_i\@$ at the corners, reproduce the content of Eqs.\;(\ref{3.10}b,c),
(3.11).
Recalling the previous discussion on the characterization of the class $\@\p\@$, we conclude that the quotient map $\@\varpi:T^*_\g\/(\V)\to V^*\/(\g)\@$ sets
up a 1-1 correspondence $\@\wp^*\/(\g)\leftrightarrow\p\@$.
\end{proof}

In view of Corollary \ref{Cor3.1}, all assertions of Theorem \ref{Teo3.1} may be rephrased, systematically replacing the class $\@\p\@$ with $\@\wp^*\/(\g)\@$.

For future reference we point out that, on account of Eqs.\;(\ref{3.10}b,c), (\ref{3.14}a), the component $\@p_0\/(t)\@$ satisfies the evolution equation
\begin{equation}\label{3.15}
\d p_0/dt\@=\@\d \Lagr_{|\h\g}/dt\@-\@\d p_i/dt\@\psi^{\@i}_{|\h\g}\@-\@p_i\@\d\psi^i_{|\h\g}/dt\@=\@\biggl(\de \Lagr/de t\biggr)_{\!\h\g}\@-\@p_i
\biggl(\de\psi^i/de t\biggr)_{\!\h\g}\,,
\end{equation}
formally similar to Eq.\;(\ref{3.10}b).

\smallskip
\subsection{Lagrange multipliers}\label{Sec3.2}
As an illustration of the algorithm developed in Sec. \!\ref{Sec3.1}, we now discuss a constructive approach to the \emph{Lagrange multipliers\/} method. To
this end, we consider a situation in which:
\begin{itemize}
\item the action integral involves a ``free'' Lagrangian $\@\L\@$, viewed as a function on $\@\j\V\@$, with local expression
$\@\L=\L\/(t,q^i,\dot q^i)\@$;
\item the presence of the constraints restricts the mobility of the system to a submanifold 
\linebreak
$\@\A\xrightarrow{i}\j\V\@$, implicitly represented by the system
\begin{equation}\label{3.16}
g_\s\/(t,q^i,\dot q^i)=0\,,\qquad\quad\s=1,\ldots,n-r\@,
\end{equation}
with $\@g_\s\in\F\/(\j\V)\@$ and $\@\rank\@\big\|\de(g_1\cdots\@g_{n-r})/de{(\dot q^1 \cdots\@\dot q^n)}\,\big\|=n-r\@$.
\end{itemize}

The geometric setup developed in Sec. \!3.1 is recovered regarding the intrinsic Lagrangian $\@\Lagr\@$ as the pull--back of the extrinsic one.

In coordinates, the situation is expressed by the equations
\begin{subequations}\label{3.17}
\begin{align}
& \Lagr\/(t,q^i,z^A)\,=\,\L\/(\?t,q^i,\psi^i\/(t,q^i,z^A)\@)\@,                          \\[3pt]
& g_\s\/(\?t,q^i,\psi^i(t,q^i,z^A)\@)\,=\,0\@.
\end{align}
\end{subequations}

Let $\@\E\@$ denote the product manifold $\@\V\times\R^{n-r}\!\vspace{1pt}$, with local coordinates $\@t,q^i,\l^\s$. The natural projection
$\@\E\xrightarrow{\pi_1}\V\@$ makes $\@\E\@$ into a vector bundle over $\@\V\@$ and therefore also into a fibre bundle over $\@\R\@$.

We regard $\@\E\to\R\@$ as the configuration space of a fictitious \emph{unconstrained\/} system, merge the functions $\@\L,\@g_\s\@$ into a
function\vspace{1pt} $\@\LL:=\L +\l^\s\@g_\s\@$, and adopt the latter as a (singular) Lagrangian over $\@\j\E\@$.

The content of Eqs.\;(\ref{3.17}a,b) is then summarized into the single expression
\begin{equation}\label{3.18}
\Lagr\/(t,q^i,z^A)\,=\,\LL\/(\?t,q^i,\psi^i\/(t,q^i,z^A)\?,\l^\a\@)\qquad\;\forall\;(\l^1,\ldots,\l^{n-r})\in\R^{n-r}\@.
\end{equation}
From the latter, we derive the relations
\begin{equation}\label{3.19}
\de \Lagr/de t\@=\@\de\LL/de t\@+\@\de\LL/de{\dot q^k}\;\de\psi^k/de t\;,\hskip.4cm
\de \Lagr/de{q^i}\@=\@\de\LL/de{q^i}\@+\@\de\LL/de{\dot q^k}\;\de\psi^k/de{q^i}\;,\hskip.4cm
\de \Lagr/de{z^A}\@=\@\de\LL/de{\dot q^k}\;\de\psi^k/de{z^A}\@.\;
\end{equation}

By the arbitrariness of $\@\l^1,\ldots,\l^{n-r}$, each equation (\ref{3.19}) is actually equivalent to $\@n-r+1\@$ separate relations. In particular, the last
one splits into
\begin{equation}\label{3.20}
\de \Lagr/de{z^A}\@=\@\de\L/de{\dot q^k}\;\de\psi^k/de{z^A}\;,\qquad \de g_\s/de{\dot q^k}\,\de\psi^k/de{z^A}\,=\@0\@.
\end{equation}

To any continuous, piecewise differentiable section $\@\mu:[\@t_0,t_1\@]\to\E\@$ let us now associate the action integral
\begin{equation}\label{3.21}
\I\@[\mu]:= \int_{j_1\!\@(\mu)}\LL\;dt\,:=\,\sum_{s=1}^N\,\int_{a_{s-1}}^{\@a_s}\!\biggl[\@\L\biggl(t,q^i,\d q^i/dt\biggr)\@+
\@\l^\s\@g_\s\biggl(t,q^i,\d q^i/dt\biggr)\@\biggr]\,dt
\end{equation}

We have then the following
\begin{theorem}\label{Teo3.2}
The projection $\@\E\xrightarrow{\pi_1}\V\@$ sets up a 1--1 correspondence between extremaloids of the functional (\ref{3.21}) and \emph{ordinary\/}
extremaloids of the functional~(\ref{3.1}).
\end{theorem}

\begin{proof}
According to elementary (holonomic) variational calculus, the extremaloids $\@\mu:q^i=q^i\/(t)\@, 
\linebreak 
\l^\s=\l^\s\/(t)\@$ of the functional (\ref{3.21}) are
determined by the Euler--Lagrange equations
\begin{align*}
& g_\s\/(t,q^i,\dot q^i)\,=\,0                                                          \\[3pt]
& \d/dt\,\de\LL/de{\dot q^i}\,-\,\de\LL/de{q^i}\,=\,0
\end{align*}
completed by the Erdmann--Weierstrass corner conditions, asserting the continuity of the ``momenta'' $\@\de\LL/de{\dot q^i}\@$ and of the ``Hamiltonian''
$\,\HH =\dot q^i\@ \de\LL/de{\dot q^i}\@-\@\LL\,$ along $\@\j\mu\@$.

Setting $\@\g=\pi_1\cdot\mu\@,\;p_0=-\big(\HH\big)_{\!\@\j\mu}\,,\;p_i=\de\LL/de{\dot q^i}\@$ and taking Eqs.\;(\ref{3.17}a,b), (\ref{3.18}), (\ref{3.19}) into
account, it is readily seen that the $1$--form $\@\rho\@=\@p_0\@d\/t_{|\g}+ p_i\@d\/q^i_{|\g}\@$ is continuous along $\@\g\@$ and satisfies the conditions
\begin{subequations}\label{3.22}
\begin{align}
& p_0\@+\@p_i\,\psi^i\big(t,q^i,\psi^i\big)\@=\@\Lagr\big(t,q^i,\psi^i\big)\@,                                                  \\[4pt]
& \d p_0/dt=\de \Lagr/de t-\@p_i\,\de\psi^i/de t\;,\quad\;\d p_i/dt=\de \Lagr/de{q^i}\@-\@p_k\,\de\psi^k/de{q^i}\;,
\quad\;p_k\,\de\psi^k/de{z^A}=\de \Lagr/de{z^A}\@,
\end{align}
\end{subequations}
reproducing the content of Eqs.\;(\ref{3.13}).

This proves that the class $\@\p\@$ is non--empty, thus ensuring that the section $\@\g\@$ is an ordinary extremaloid of the functional~(\ref{3.1}).

Conversely, if $\@\g\@$ is an ordinary extremaloid, the condition $\@\p\ne\emptyset\@$ ensures the existence of at least a $1$--form
\mbox{$\,\rho=p_0\@d\/t_{|\g}\@+p_i\@ d\/q^{\@i}\@$} satisfying Eqs.\;(\ref{3.22}a,b).

From these, taking Eqs.\;(\ref{3.20}), the functional independence of the $\@g_\a$'s and the Rouch\'e--Capelli theorem into account, we conclude that the linear
system
\begin{equation}\label{3.23}
p_i\/(t)\@=\@\biggl(\de\L/de{\dot q^i}\biggr){\!\plus07}_{j_1\!\@(\g)}\!+\@\l^\a\@\biggl(\de g_\s/de{\dot q^k}\biggr){\!\plus07}_{j_1\!\@(\g)}\@=\@
\biggl(\de\LL/de{\dot q^i}\biggr){\!\plus07}_{j_1\!\@(\g)}
\end{equation}
admits a unique solution $\@\l^\a=\l^\a\/(t)\@$, $\@\a=1,\ldots,n-r\@$. The pair $\@(\g,\rho)\@$ determines therefore a unique section
$\@\mu:[\@t_0,t_1\@]\to\E\@$, satisfying $\@\pi_1\cdot\mu=\g\@$.

On the other hand, through a straightforward pull--back procedure, the $1$--form $\@\rho\@$ induces a continuous $1$--form
$\@\h\rho=p_0\,d\/t_{|\mu}\@+p_i\,d\/q^{\@i}_{|\mu}\@$ along $\@\mu\@$.

On account of Eqs.\;(\ref{3.18}), (\ref{3.19}), (\ref{3.22}a,b), (\ref{3.23}), denoting by $\@\dot\mu\@$ the tangent vector field to $\@\mu\@$, the latter
satisfies the relations
\begin{align*}
& \big<\h\rho\@,\@\dot\mu\?\big>\@=\@p_0\@+p_i\@\dot q^{\@i}_{\,|\j\mu}\@=\@p_0\@+\@p_i\,\psi^i_{\,|\h\g}\@=\@\Lagr_{|\h\g}\@=\@\LL_{|\j\mu}\@,        \\[2pt]
&\d p_0/dt\@=\@\biggl(\de \Lagr/de t\biggr){\!\plus07}_{\!\@\h\g}-\@p_i\@\biggl(\de\psi^i/de t\biggr){\!\plus07}_{\!\@\h\g}\@=
\@\biggl(\de \Lagr/de t\biggr){\!\plus07}_{\!\@\h\g}-\biggl(\de\LL/de{\dot q^i}\biggr){\!\plus07}_{j_1\!\@(\g)}\@\biggl(\de\psi^i/de t\biggr)_{\!\@\h\g}\@=
\@\biggl(\de\LL/de t\biggr){\!\plus07}_{j_1\!\@(\g)} \@,                                                                                                \\[2pt]
&\d p_i/dt\@=\@\biggl(\de \Lagr/de{q^i}\biggr){\!\plus07}_{\!\@\h\g}-\@p_k\@\biggl(\de\psi^k/de{q^i}\biggr){\!\plus07}_{\!\@\h\g}\@=
\@\biggl(\de \Lagr/de{q^i}\biggr){\!\plus07}_{\!\@\h\g}-\biggl(\de\LL/de{\dot q^k}\biggr){\!\plus07}_{j_1\!\@(\g)}\@\biggl(\de\psi^k/de{q^i}\biggr)_{\!\@\h\g}\@=
\@\biggl(\de\LL/de{q^i}\biggr){\!\plus07}_{j_1\!\@(\g)}\@,
\end{align*}
whence also, making use of the identity $\@\big(\de\LL\,/de{\l^\a}\big){\!\plus03}_{j_1\!\@(\g)}=0\@$
\begin{equation*}
\d\h\rho/dt\,=\,\d p_0/dt\,d\/t_{\,|j_1\!\@(\g)}\,+\,\d p_i/dt\,d\/q^i_{\,|j_1\!\@(\g)}\,+\,p_i\,d\/\dot q^i_{\,|j_1\!\@(\g)}\,=\,d\LL_{\,|j_1\!\@(\g)}
\end{equation*}

This proves that the class $\@\wp\/(\mu)\@$ is non--empty, thus ensuring that the section $\@\mu\@$ is an extremaloid of the functional~(\ref{3.21}).
\end{proof}

From the proof of Theorem \ref{Teo3.2} one can easily derive the following
\begin{corollary}\label{Cor3.2}
An extremaloid $\@\g\@$ of the functional (\ref{3.1}) is \emph{normal\/} if and only if there exists \emph{exactly one\/} extremaloid $\@\mu\@$ of the
functional~(\ref{3.21}) satisfying $\@\g=\pi_1\cdot\mu\@$.
\end{corollary}

\smallskip
\subsection{Pontryagin's equations in $\@\C\/(\A)\@$}\label{Sec3.3}
The algorithm (\ref{3.10}a,b,c) involves a set of $\@2\@n+r\/$ equations for the unknowns $q^i\/(t)\@,\@z^A\/(t)\@,\@p_i\/(t)\@$. As such, it has a natural
setting in the geometrical environment provided by the contact bundle $\@\C\/(\A)\@$.

As pointed out in Sec. \!\ref{Sec2.2}, the latter is a vector bundle over $\@\A\@$, referred to fibred coordinates $\@t,q^i,z^A\!\?,\?p_i\,$, and identified
with the pull--back of the space $\@V^*\/(\V)\,$ through the commutative diagram
\begin{equation*}
\begin{CD}
\C\/(\A )         @>\k>>          V^*\/(\V)        \\
@V{\z}VV                          @VV\pi V         \\
\A                @>\pi>>           \V
\end{CD}
\end{equation*}

The advantage of the environment $\@\C\/(\A)\,$ comes from the presence of the Liouville $1$--form (\ref{2.6}).
By means of the latter, the Lagrangian $\@\Lagr\in\F\/(\A)\,$ may be lifted to a $1$--form $\@\thL\@$ over $\@\C\/(\A)\@$ according to the prescription\@%
\footnote
{\@A deeper insight into the geometrical meaning of the $1$--form (\ref{3.24}) comes from the study of the gauge--theoretical structure of the
calculus of variations, as developed in \cite{BLP}.}%
\footnote%
{\@As usual, we use the same symbol for covariant objects in $\@\A\@$ and for their pull--back in $\@\C\/(\A)\@$.}
\begin{equation}\label{3.24}
\thL\,:=\,\Lagr\,d\/t\@+\@\Th\,=\,\left(\Lagr-p_i\,\psi^i\right)d\/t\@+\@p_i\,d\/q^i\,:=\,
-\,\Ham\,d\/t\@+\@p_i\,d\/q^i
\end{equation}
The function $\@\Ham\/(t,q^i\@\!,p_i,z^A):=p_i\,\psi^i-\Lagr\@$ is called the \emph{Pontryagin Hamiltonian\/}. The $1$--form (\ref{3.24}) is called the
\emph{Pontryagin--Poincar\'e--Cartan\/} $1$--form.

To understand the role of $\@\thL\@$, we focus on the fibration $\@\C\/(\A)\xrightarrow{\u}\V\@$ given by the composite map $\@\u:=\pi\cdot\k\@$. A piecewise
differentiable section $\@\arc{\@\Chi\@}{\@t_0,t_1}\@$ consisting of a finite family of closed arcs
\begin{equation*}
{\Chi}^{\ns}:[\@a_{s-1},a_s\@]\to\C\/(\A)\,,\qquad s=1,\ldots,N,\quad t_0=a_0<a_1<\cdots<a_N=t_1
\end{equation*}
is called $\@\u$--{\it continuous\/} if and only if it projects onto a continuous, piecewise differentiable section
\linebreak
$\@\u\cdot\Chi:[\@t_0,t_1\@]\to\V\@$.\vspace{1pt}

A deformation $\@\Chi_{\?\xi}=\big\{\arc{\@\Chi_{\;\xi}^{\ns}}{\,a_{s-1}\/(\xi),a_s\/(\xi)}\@\big\}\@$  is called $\@\u$--continuous if and only if all
sections $\@\Chi_{\?\xi}\@$ are $\@\u$--continuous. The necessary and sufficient condition for this to happen is the validity of the matching conditions
\begin{equation}\label{3.25}
\u\cdot\Chi^{\ns}_{\;\xi}\/(a_s\/(\xi))\,=\,\u\cdot\Chi^{\nS}_{\;\xi}\/(a_s\/(\xi))\qquad\;\forall\;|\?\xi\?|<\eps\,,\;s=1,\ldots,N-1\@.
\end{equation}

A deformation $\@\Chi_{\?\xi}\@$ is said to preserve the endpoints of $\@\u\cdot\Chi\@$ if and only if the projection $\@\u\cdot\Chi_{\?\xi}\@$ does so. A
vector field along $\@\Chi\@$ tangent to the orbits of a \mbox{$\@\u$--continuous} deformation is called an \emph{infinitesimal deformation\/}.

It is worth noticing that the previous definitions do not require the admissibility of the sections $\,\u\cdot\Chi\,$ or $\,\u\cdot\Chi_{\?\xi}\@$\vspace{1pt}.
In this respect, the only condition needed in order for a vector field $\@X^i\@\big(\de/de{q^i}\big)_\chi +\@\G^A\big(\de/de{z^A}\big)_\chi
+\@\Pi_{\?i}\?\big(\de/de{p_i}\big)_\chi\@$\vspace{1pt} to represent an infinitesimal deformation of $\@\Chi\@$ is the consistency with the matching conditions
(\ref{3.25}), summarized into the \emph{jump relations\/}
\begin{equation}\label{3.26}
\biggl[\@X^i+\@\a_s\,\d q^i/dt\@\biggr]_{a_s}=\,0\qquad s=1,\ldots,N-1
\end{equation}
with $\@\a_s=\big(\d a_s/d\xi\big)_{\xi=0}\,$ and with no a--priori relationship between $\@\d q^i/dt\@$ and $\@\psi^i_{|\chi}\@$.

\medskip
By means of $\@\thL\@$ we define an action integral $\,\I\@[\@\Chi\@]\@$ over $\@\C\/(\A)\/$, assigning to each \mbox{$\@\u$--continuous} section
$\@\Chi:\@q^i=q^i\/(t)\@,\;z^A=z^A\/(t)\@,\;p_i=p_i\/(t)\@$ the real number
\begin{equation}\label{3.27}
\hskip-0cm\I\@[\@\Chi\@]:=\!\int_{\chi}\thL=\int_{t_0}^{t_1}\!\!\!\left(p_i\,\d q^i/dt -\Ham\!\@\right)\!d\/t=
\sum_{s=1}^N\int_{a_{s-1}}^{a_s}\!\!\!\left(p^{\ns}_{\@i}\,\d q^{\@i}_{\ns}/dt -\Ham{\plus03}_{\!\@|\Chi^{\!\ns}}\!\right)\!d\/t\@.
\end{equation}

For any $\u$--continuous deformations $\@\Chi_{\?\xi}\@$\vspace{2pt} preserving the endpoints of \mbox{$\@\u\cdot\Chi\@$,} denoting by
$\@X^i\@\big(\de/de{q^i}\big)_\chi +\@\G^A\big(\de/de{z^A}\big)_\chi +\@\Pi_{\?i}\?\big(\de/de{p_i}\big)_\chi\@$\vspace{2pt} the corresponding infinitesimal
deformation, we have then the relation
\begin{align*}
\d\@\I\@[\@\Chi_{\?\xi}\@]/d\xi\,\bigg|_{\@\xi=0}\,=\@&\int_{t_0}^{t_1} \left[\bigg(\@\d\?q^i/dt\@-\@\de\Ham/de{p_i}\bigg)\@\Pi_{\?i}\@-
\@\bigg(\@\d\?p_i/dt\@+\@\de\Ham/de{q^i}\bigg)\@X^i\@-\@\de\Ham/de{z^A}\,\G^A\@\right]d\/t\;+                                                   \\[3pt]
&-\@\sum_{s=1}^{N-1}\biggl[\@p_i\biggl(X^i+\@\a_s\,\d q^i/dt\@\biggr)-\,\a_s\,\Ham\biggl]_{a_s}\@.
\end{align*}

From this, taking Eqs.\;(\ref{3.26}) into account\vspace{1pt}, we conclude that the vanishing of $\@\d\@\I/d\xi\big|_{\@\xi=0}\@$ for arbitrary
$\@\Chi_{\?\xi}\@$ is mathematically equivalent to the system
\begin{subequations}\label{3.28}
\begin{align}
& \d\/q^i/d t\,=\,\de\?\Ham/de{p_i}\,\@,\qquad \d\/p_i/d t\,=\,-\@\de\?\Ham/de{q^i}\,\@,\qquad \de\?\Ham/de{z^A}\,=\,0\,\@,                      \\
\intertext{completed with the continuity conditions}
& \big[\@p_i\@\big]_{\?a_s}\,=\,\big[\@\Ham\big]_{\?a_s}\,=\,0\qquad\quad s=1,\ldots,N-1\@.
\end{align}
\end{subequations}

In view of the definition of the Pontryagin Hamiltonian $\@\Ham\@$, Eqs.\;(\Ref{3.28}a,b) are easily seen to reproduce the content of Eqs.\;(\Ref{3.10}a,b,c),
(\ref{3.11})\@
\footnote%
{Actually, exactly as it happened in Sec. \!3.1, the condition $\@\big[\@q^i\@\big]_{\?a_s}=0\@$ does not arise from the stationarity requirement, but is
implicit in the definition of $\@\Chi\@$.}.

As far as the \emph{ordinary\/} extremaloids are concerned, the original (constrained) variational problem in the configuration space is therefore equivalent
to a \emph{free\/} variational problem in the contact manifold, in full agreement with Pontryagin's \emph{maximum principle\/}.

As a further comment on Eqs.\;(\ref{3.28}), let us digress on the special situation determined by the ansatz $\@\Lagr=0\@$. In the stated circumstance, the
functional
\begin{equation}\label{3.29}
\I_0[\@\Chi\@]\,:=\,\int_{\chi}\?\Th\,=\,\int_{t_0}^{t_1}\?p_i\!\left(\d q^i/dt -\psi^i\right)d\/t
\end{equation}
is an intrinsic attribute of the manifold $\@\C\/(\A)\@$, entirely determined by the Liouville $1$--form (\ref{2.6}). Its role is clarified by the following
\begin{proposition}\label{Pro3.1}
Let $\@\g:[\@t_0,t_1\@]\to\V$ denote any continuous, piecewise differentiable section. Then: \ALPH
\begin{enumerate}
\item
$\@\g$ is admissible if and only if the functional (\ref{3.29}) admits at least one extremaloid $\@\Chi\@$ projecting onto $\@\g\@$, i.e.~satisfying
$\@\u\cdot\Chi=\g\@$;
\item
for any admissible $\/\g\?$, the extremaloids of $\,\I_{\@0}\@$ projecting onto $\?\g\?$ are in 1-1 correspondence with the elements of the annihilator
$\?\big(\@\U\?(\W)\@\big){}^0$.
\end{enumerate}
\end{proposition}
\begin{proof}
For $\@\Lagr=0\@$, Eqs.\;(\Ref{3.28}a,b) reduce to

\vskip-12pt
\begin{subequations}\label{3.30}
\begin{align}
& \d q^i/d t\,=\,\psi^i\/(t,q^i,z^A)\,,                                                                 \\[2pt]
& \d p_i/d t\,+\,\de\psi^k/de{q^i}\,p_k\,=0\,,                                                          \\[0pt]
& p_i\,\de\psi^i/de{z^A}\,=\,0\,,                                                                       \\[4pt]
\big[&\@p_i\@\big]_{\?a_s}\,=\,\big[\@p_i\,\psi^i\@\big]_{\?a_s}\,=\,0\@,\qquad\quad s=1,\ldots,N-1\@.
\end{align}
\end{subequations}

Eq.~(\Ref{3.30}a) is the admissibility requirement for the section $\@\u\cdot\Chi\@$. Accordingly, if an extremaloid $\@\Chi\@$ of the functional (\ref{3.29})
projects onto $\@\g\@$, its image $\@\z\cdot\Chi\@$ under the map $\@\z:\C\/(\A)\to\A\@$ coincides with the lift $\@\h\g:[\@t_0,t_1\@]\to\A\@$.

For any admissible $\@\g\@$, the extremaloids $\@\Chi\@$ projecting onto $\@\g\@$ are therefore in 1--1 correspondence with the solutions $\@p_i\/(t)\@$ of the
homogeneous system (\Ref{3.30}b,c,d), with the functions $\@q^i\/(t),\@z^A\/(t)\@$ regarded as given.

Assertions \textit{a)} and \textit{b)} are then straightforward consequences of the fact that, according to Proposition \ref{Pro2.4}\@, Eqs.\;(\Ref{3.30}b,c,d)
are precisely the conditions required in order for a virtual $\@1$--forms $\@p_i\/(t)\,\wg i\@$ to belong to the annihilator $\?\big(\@\U\?(\W)\@\big){}^0\?$.
\end{proof}
According to Proposition \ref{Pro3.1}, a section $\g:[\?t_0,t_1\?]\to\V$ describes a \emph{normal\/} evolution of the system if and only if the functional
(\ref{3.29}) admits \emph{exactly one\/} extremaloid projecting onto~$\g$, namely the one corresponding to the trivial solution $\@p_i\/(t)=0\@$. If the
extremaloids projecting onto~$\g$ are more than one, $\@\g\@$ represents an \emph{abnormal\/} evolution; if no such extremaloid exists, $\@\g\@$ is not
admissible.

Returning to the action integral (\ref{3.27}) we can now state
\begin{corollary}\label{Cor3.3}
The totality of extremaloids of the functional (\ref{3.27}) projecting onto a section 
\linebreak
$\?\g:[\@t_0,t_1\@]\to\V$ is an affine space, modelled on the vector
space formed by the extremaloids of the functional (\ref{3.29}) projecting onto~$\g\?$.
\end{corollary}
The proof, entirely straightforward, is left to the reader.

The previous arguments provide a restatement of Theorem \ref{Teo3.1} in the environment $\@\C\/(\A)\@$. As already pointed out, the projection algorithm
$\@\Chi\to\u\cdot\Chi\?$ does not reproduce \emph{all\/} the extremaloids of the functional (\ref{3.1}), but only the \emph{ordinary\/} ones.

The missing ones may be obtained determining the abnormal evolutions by means of Proposition \ref{Pro3.1}, finding out which ones have an exceptional
character, and directly testing the validity of Eq.~(\ref{3.2}a) on each of them.

\smallskip
\subsection{Hamiltonian formulation}\label{Sec3.4}
As pointed out in Sec. \!3.3, all ordinary extremaloids of the functional (\ref{3.1}) are projections of extremaloids of the functional (\ref{3.27}). Let us
discuss the implications of this fact.

To this end, temporarily leaving aside all aspects related to the presence of corners, we observe that a curve $\?\Chi\@$ in $\@\C\/(\A)\@$ is an extremal for
the functional (\ref{3.27}) if and only if its tangent vector field $\@Z:=\Chi_*\big(\de/de t\big)\@$ satisfies the properties
\begin{equation}\label{3.31}
\big<Z\@,\@d\/t\big>=1\,,\qquad Z\@\internal d\@\thL=0\@.
\end{equation}

On account of Eq.\;(\ref{3.24}), at each $\@\vs\in\C\/(\A)\@$, a necessary and sufficient condition for the existence of at least one vector $\@Z\in
T_\vs\/(\C\/(\A))\@$ satisfying Eqs.\;(\ref{3.31}) is the validity of the relations
\begin{subequations}\label{3.32}
\begin{equation}
\bigg(\de\@\Ham/de{z^A}\bigg)_{\!\vs}\@=\@0\,, \quad\qquad\;A=1,\ldots, r\@.
\end{equation}

Points $\@\vs\@$ at which Eqs.\;(\ref{3.31}) admit a \emph{unique\/} solution $Z$ will be called \emph{regular points\/} for the functional (\ref{3.27}). In
coordinates, the regularity requirement is expressed by the condition
\begin{equation}
\det\bigg(\SD\Ham/de{z^A}/de{z^B}\bigg)_{\!\vs}\,\ne\@0\@.
\end{equation}
\end{subequations}

In view of Eq.\;(\Ref{3.32}b), in a neighborhood of each regular point Eqs.\;(\ref{3.32}a) may be solved for the $\@z^A\@$'s, giving rise to a representation
of the form
\begin{equation}\label{3.33}
z^A\,=\,z^A\@(t,q^1\?,\ldots,q^n\?,\@p_1,\ldots,p_n)\@.
\end{equation}
The regular points form therefore a $\?(2n+1)$--dimensional \vspace{-2pt} submanifold $\@\SS\xrightarrow{j}\C\/(\A)\@$, locally diffeomorphic to the manifold
$\@V^*\/(\V)\@$.

Inserting Eqs.\;(\ref{3.33}) into the first pair of relations (\Ref{3.28}a) gives rise to a system of ordinary differential equations in normal form for the
unknowns $\@q^i\/(t),p\?_i\/(t)\@$. The algorithm is readily implemented denoting by $\@\H:=j^*(\Ham)\@$ the pull-back of the Pontryagin Hamiltonian, expressed
in coordinates as
\begin{equation*}
\Ham\/(t,q^i\!,z^A(t,q^i\!,\?p\?_i)\@,\?p\?_i)=p\?_k\,\psi^k\/(t,q^i\!,z^A(t,q^i,\?p\?_i))-\Lagr\/(t,q^i\!,z^A(t,q^i,\?p\?_i))
\end{equation*}

In view of Eqs.\;(\Ref{3.32}a) we have the identifications
\begin{align*}
&\de \H/de{p\?_i}\,=\,\de\Ham/de{p_i}\@+\@\cancel{\de\Ham/de{z^A}}\,\de z^A/de{p_i}\,=\,\psi^i\,,                                                   \\[3pt]
&\de \H/de{q\?^i}\,=\,\de\Ham/de{q\?^i}\@+\@\cancel{\de\Ham/de{z^A}}\,\de z^A/de{q^i}\,=\,p_k\,\de\psi^k/de{q\?^i}\,-\,\de \Lagr/de{q\?^i}\@.
\end{align*}

\noindent
On account of these, the first pair of equations (\Ref{3.28}a) takes the form
\begin{subequations}\label{3.34}
\begin{alignat}{2}
 &\d\/q^i/d t\,=\,&&\de \H/de{p_i}\,, \hskip5cm                              \\[2pt]
 &\d\/p_i/d t\,=\,&-&\de \H/de{q^i}\,.
\end{alignat}
\end{subequations}

The original \emph{constrained\/} lagrangian variational problem is thus converted into a \emph{free\/} hamiltonian problem on the submanifold
$\@j:\SS\to\C\/(\A)\@$.

A $\@\u$--continuous extremaloid of the functional (\ref{3.27}) consisting of a finite family of closed arcs
$\@\Chi^{\narc{s}}:[\@a_{s-1},a_s\@]\to\C\/(\A)\@$, each contained in (\/a connected component of\/) the submanifold $\@\SS\@$ will be called a \emph{regular
extremaloid\/}.

Singular extremaloids, partly, or even totally lying outside $\SS\/$ may also exist. In fact, while Eq.\;(\Ref{3.32}a) is part of the system (\Ref{3.28}a,b),
and must therefore be satisfied by any extremaloid, the requirement (\Ref{3.32}b) has only to do with the \emph{well--posedness\/} of the Cauchy problem for
the subsystem (\Ref{3.28}a).

On the other hand, by construction, the Hamilton equations (\Ref{3.34}a,b) determines only the regular extremaloids. The singular ones, when present, have
therefore to be dealt with directly, looking for solutions of Eqs.\;(\Ref{3.28}a,b) not arising from a well--posed Cauchy problem.
In principle, this could be done extending to the non--holonomic context the concepts and methods commonly adopted in the study of singular Lagrangians
\cite{Gotay}. The argument is beyond the purposes of the present work, and will not be pursued.

\smallskip

\appendix
\section{Finite deformations with fixed end points: an existence theorem}\label{SecA}
\textbf{(\/i\/)} \, Given an admissible, piecewise differentiable section $\g:[\?t_0,t_1\?]\to\V$, a crucial question is establishing under what circumstances
\emph{every\/} admissible infinitesimal deformation vanishing at the endpoints of $\@\g\@$ is tangent to an admissible finite deformation $\@\g_{\@\xi}\@$ with
fixed endpoints. The following preliminaries help simplifying the discussion.

\begin{proposition}\label{ProA.1}
Let $\@\h\g:(c,d)\to\A\@$ be the lift of an admissible differentiable section $\@\g:(c,d)\to\V\,$. Then, for any closed interval $\@[\?a,b\?]\subset(c,d)\@$
there exists a fibred local chart $\@(U,k)\@,\,k=(t,q^1,\ldots,q^n,z^1,\ldots,z^r)\,$ satisfying the properties
\begin{subequations}\label{A.1}
\begin{align}
&\bullet\;\, \h\g\/(t)\in U\quad\forall\;t\in[\?a,b\@]\@;                                                                                  \\[1pt]
&\bullet\;\,\mbox{\it the intersection $\@\h\g\big(\?(c,d)\?\big)\cap U\/$ coincides with the curve $q^i=z^A=0\@$;} \hskip.4cm             \\[-2pt]
&\bullet\;\,\psi^i\?\big(\h\g\/(t)\?\big)\,=\,\Big(\?\De\psi^i/de{q^k}\?\BIG)_{\h\g\/(t)} =\,0
\qquad\forall\;\h\g\/(t)\in U\@.
\end{align}
\end{subequations}
\end{proposition}
\begin{proof}
The existence of fibred local charts $\@(U,\@\h k)\@,\;\h k=(t,\h q\@^1,\cdots\,\h q\@^n,\h z^1,\cdots\,\h z^r)\@$ satisfying Eqs.\;(\ref{A.1}\@a,\,b) and the
first relation (\ref{A.1}\@c) is entirely straightforward.

Choose any such chart, and denote by $\@\h{\dot q}\@^i=\h\psi\@^i\/(t,\h q\@^i,\h z^A)\@$ the corresponding representation of the imbedding $\@\A\to\j\V\@$.

Under an arbitrary transformation $\@q^i=\a\?^i{}_j\/(t)\,\h q\@^j\!\@,\;\@z^A=\h z^A\@$ we have then the transformation laws
\begin{equation*}
\psi\?^i\@=\@\d\@\a\?^i{}_j/d t\;\h q\@^j\@+\@\a\?^i{}_j\,\h\psi\@\?^j,\qquad\;
\de\psi\?^i/de{q^k}\,=\,\left(\d\@\a\?^i{}_j/d t\,+
\,\a\?^i{}_r\,\de\@\h\psi\@^r/de{\h q\@^j}\right)\big(\@\a^{-1}\big)^{\?j}{}_k\,.
\end{equation*}
Therefore, if the matrix $\@\a^i{}_j\/(t)\@$ obeys the transport law
\begin{equation*}
\d\@\a^i{}_j/d t\,+\,\a^i{}_r\@\bigg(\de\@\h\psi^r/de{\h q\?^j}\bigg)_{\!\h\g\/(t)}=\ 0\@,
\end{equation*}
the coordinates $\@t,q^i,z^A\@$ satisfy all stated requirements.
\end{proof}

Every local chart $\@(U,k)\@$ consistent with Eqs.\;(\ref{A.1}\@a,\,b,\,c) \@is said to be \emph{adapted\/} to the arc $\@\arc{\h\g}{\?a,b\?}\@$. For later use
we observe that this property is stable under arbitrary transformations of the form
\begin{equation}\label{A.2}
\bar q^i\@=\@\bar q^i\/(q^1,\ldots,q^n)\,,\qquad\bar z^A=\bar z^A\?(t,q^1,\ldots,q^n,z^1,\ldots,z^r)\@,
\end{equation}
with $\@\bar q^i\/(0,\ldots,0)\@=\@\bar z^A\?(t,0,\ldots,0,\ldots,0)\@=\@0\@$.
\begin{corollary}\label{CorA.1}
Let $\@\h\g=\big\{\arc{\@\h\g^{\ns}}{\@a_{s-1},a_s}\@,\;s=1,\ldots,N\@\big\}\vspace{2pt}\@$ be the lift of an admissible piecewise differentiable section
\vspace{1pt} $\arc\g{\@t_0,t_1\?}\@$. Then, there exist fibred local charts $\@(U_s,k_s)\@$\vspace{1pt},
$\@k_s=\big(\?t,q_{\ns}^{\;1}\@,\ldots,q_{\ns}^{\;n}\@, z_{\ns}^{\;1},\ldots,z_{\ns}^{\;r}\@\big)\@$ adapted to the arcs $\@\h\g^{\ns}\@$ such that, in each
intersection $\@\pi\/(U_s)\cap\pi\/(U_{s+1})\@$, the coordinate transformations
$\@q_{\nS}^{\;i}=q_{\nS}^{\;i}\/(t,q_{\ns}^{\;1}\@,\ldots,q_{\ns}^{\;n}\@)\/$ satisfy the conditions
\begin{equation}\label{A.3}
\biggl(\de q_{\nS}^{\;i}/de{q_{\ns}^{\;j}}\biggr)_{\!\g\/(a_s)}=\;\,\delta\@^i_j\;;\qquad
\biggl(\de q_{\nS}^{\;i}/de t\biggr)_{\!\g\/(a_s)}=\,\@\,-\@\Big[\?\psi^i\/(\h\g)\?\Big]_{\?a_s}\@.
\end{equation}
\end{corollary}
\begin{proof}
The conclusion follows at once from Proposition \ref{ProA.1}, observing that the freedom in the choice of the the adapted coordinates summarized into
Eqs.\;(\ref{A.2}) leaves full control on the values of the Jacobians $\@\Big(\de q_{\NS}^{\;i}/de{q_{\Ns}^{\;j}}\Big){\plus05}_{\!\g\/(a_s)}\!\@$.
In particular, in the intersection $\@\pi\/(U_{s-1}\cap U_s)\@$, the arcs $\@\g^{\ns}\@$, $\@\g^{\nS}\@$ are respectively described by the equations
\begin{equation*}
\left\{
\begin{aligned}
&q_{\nS}^{\;i}\/(\g^{\ns})\@=\@q_{\nS}^{\;i}\big(t, q_{\ns}^{\;1}\/(\g^{\ns}),\ldots,q_{\ns}^{\;n}\/(\g^{\ns})\big)\@=\@q_{\nS}^{\;i}\/(t,0,\ldots,0)\@, \\[2pt]
& q_{\nS}^{\;i}\/(\g^{\nS})\@=\@0\@.
\end{aligned}
\right.
\end{equation*}

In the coordinate system $\@t,q_{\nS}^{\;i}\@$, the jump of the tangent vector at the corner $\@\g\/(a_s)\@$ is therefore given by
\begin{equation*}
\Big[\?\psi^i\/(\h\g)\?\Big]_{\?a_s}\@=\@\biggl(\d q_{\nS}^{\;i}\/(\g^{\nS})/dt\@-\@\d q_{\nS}^{\;i}\/(\g^{\ns})/dt\BIGGR)_{\!\g\/(a_s)}\@=\@-\@
\biggl(\de q_{\nS}^{\;i}/de t\BIGGR)_{\!\g\/(a_s)}\@.\hskip1cm
\end{equation*}
\end{proof}

Every family of local charts $\!\big\{\?(U_s,k_s)\@,\,s=1,\ldots,N\@\big\}\@$ satisfying the requirements of Corollary \ref{CorA.1} is said to be
\emph{adapted\/} to the section $\@\h\g\@$.

\medskip\noindent
\textbf{(\/ii\/)} \,Assigning an adapted family of local charts singles out a distinguished class of controls 
\linebreak
$\@\s^{\ns}\!:\pi\/(U_s)\to U_s\@$, described in
coordinates as
\begin{equation}\label{A.4}
\s^{\ns}\!:\quad z_{\ns}^{\;A}\/(t,q_{\ns}^{\;1},\ldots,q_{\ns}^{\;n})\,=\,0\@.
\end{equation}

By construction, each such $\@\s^{\ns}\@$ satisfies $\,\s^{\ns}\cdot \g^{\ns}=\h\g^{\ns}$, i.e.\ it contains the section $\@\g^{\ns}\@$ in the sense
described in Sec. \!\ref{Sec2.1}.

For each $\@s=1,\ldots,N\@$, the restriction of the tangent map $\@\s^{\ns}_{*}\@$ to the vertical bundle along $\@\g^{\ns}\@$ determines an infinitesimal
control $\@h^{\ns}:V\?(\g^{\ns})\to\A\?(\g^{\ns})\@$, expressed in coordinates as\vspace{2pt}
\begin{equation*}
h^{\ns}\biggl(\de/de{q_{\ns}^{\;i}}\biggr)_{\!\g^{\ns}\/(t)}\!=\;\s^{\ns}_{*}\biggl(\de/de{q_{\ns}^{\;i}}\biggr)_{\!\g^{\ns}\/(t)}\!=\;
\biggl(\de/de{q_{\ns}^{\;i}}\biggr)_{\!\h\g^{\ns}\/(t)}\;\,\Longleftrightarrow\quad\; h_i{}^A\/(t)\@=\@0\@.\vspace{2pt}
\end{equation*}

In view of Eqs.\;(\ref{2.24}\@b), (\ref{2.25}a,b), (\ref{A.1}\@c), the absolute time derivative associated with $\@h^{\ns}\@$ satisfies
\begin{equation}\label{A.5}
\D/Dt\biggl(\de/de{q_{\ns}^{\;i}}\biggr)_{\!\g^{\ns}\/(t)}\!=\;\D/Dt\;\h\w^i{}_{|\?\g^{\ns}\/(t)}\,=\,0\@,\hskip1.6cm s=1,\ldots,N\@.
\end{equation}

Noting that, according to Eq.\;(\ref{A.3}), the fields $\Big(\de/de{q_{\ns}^{\;i}}\Big){\plus05}_{\!\@\g^{\ns}}$\vspace{-1pt} --- and therefore also the
virtual $1$--forms $\wgs i$\vspace{2pt} --- match continuously at the corners, we conclude that the sections \mbox{$e_{(i)}:[\?t_0,t_1\@]\to \Vg\@$},
$\,e^{(i)}:[\?t_0,t_1\@]\to V^*\?(\g)\@$ respectively defined by
\begin{equation}\label{A.6}
e_{(i)}\/(t)\/=\biggl(\de/de{q_{\ns}^{\;i}}\biggr)_{\!\g^{\ns}\/(t)}\,,\quad e^{(i)}\/(t)\@=\@\@\delta q^i{\!\@\plus03}_{|\g^{\ns}(t)}\,,\quad s=1,\ldots,N
\end{equation}
form a dual bases for the spaces $\@V_h\@$, $V^*_h\@$ of $\?h$--transported fields along $\@\g\@$.

\medskip\noindent
\textbf{(\/iii\/)} \,Let us now come to the main question. Let $\?\g:=\big\{\arc{\@\g^{\ns}}{\@a_{s-1},a_s}\@\big\}$ be an admissible piecewise differentiable
section, $\?\{(U_s,k_s)\}\@$ a family of local charts adapted to the lift $\?\h\g\@$, and $\?\{\?e_{(i)}\?\}\@$, $\{\?e^{(i)}\?\}\@$ the corresponding
$\?h$--transported  bases.

We recall that, with the notation of Sec. \!\ref{Sec2.7}, the most general infinitesimal deformation $\@X\@$ of $\@\g\@$ is determined, up to initial data, by
an element $\@(U,\atilde)\in\W\@$, namely by a vertical vector field $\@U\@$ along $\?\h\g\?$ and by a collection of real numbers
$\@\atilde=(\a_1,\ldots,\a_{N-1})\@$.  In particular, a necessary and sufficient condition for $\?X\?$ to vanish at the endpoints of $\@\g\@$ is expressed by
the requirement $\@\U\?(U,\atilde)=0\@$ which, in adapted coordinates, reads
\begin{equation}\label{A.7}
\int_{t_0}^{t_1}\@U^A\@\bigg(\de\psi^i/de{z^A}\bigg)_{\h\g}\,d\/t\,-
\,\sum_{s=1}^{N-1}\a_s\@\big[\?\psi^i(\h\g)\?\big]_{a_s}\,=\,0\@.
\end{equation}

\smallskip
To inquire whether a given infinitesimal deformation vanishing at the endpoints of $\@\g\@$ is tangent to a finite deformation with fixed endpoints we
introduce an auxiliary tool, namely a positive--definite scalar product in $\@V\/(\A)\@$.

From a tensorial viewpoint, this means assigning a contravariant field over $\@\A\@$, described in fibred coordinates as
$\@G=G^{AB}\@\de/de{z^A}\otimes\de/de{z^B}\@$, with the components $\@G^{AB}\@$ identifying a symmetric, positive--definite matrix, inverse of the matrix
$\@G_{AB}\@$ formed by the scalar products $\@\big(\de/de{z^A}\@,\@\de/de{z^B}\big)\@$.\vspace{2pt}

Given any $\@(U,\atilde)\in\W\@$, we shall now construct a family of finite deformations $\@\g_{(\xi\?,\@\nu)}\@$ of the original section $\@\g\@$, tangent to
the infinitesimal deformation determined by $\@(U,\atilde)\@$ and depending on $\@n\@$ auxiliary parameters, identified with the components of a
$\@h$--transported $1$--form $\@\nu\in V^*_h\@$. To this end, we introduce
\begin{itemize}
\item
a family of functions
\begin{subequations}\label{A.8}
\begin{equation}
a_s\/(\xi,\nu)\,:=\,a_s\,+\,\a_s\,\xi\,-\,\tfrac12\;\a_s^{\,2}\,\xi^{\,2}\@\big[\?\psi^{\?i}(\h\g)\?\big]_{\?a_s}\,\nu_i\,,       \qquad\;s=1,\ldots,N-1\@,
\end{equation}
expressing the deformation of the partition of the interval $\@[\?t_0,t_1]\@$ into subintervals $\@[\?a_{s-1},a_s]\@$. For notational convenience, the
quantities (\ref{A.8}a) are completed by the constant functions 
\linebreak
$\@a_0\/(\xi,\nu)=t_0\,,\;a_N\/(\xi,\nu)=t_1\@$;

\smallskip
\item
a family of deformations $\@\s^{\,\ns}_{\!(\xi\?,\@\nu)}:\pi\/(U_s)\to U_s\@$ of the controls (\ref{A.4}), described in adapted coordinates as
\footnote%
{Eq.~(\ref{A.8}b) is strictly coordinate--dependent. As such, it has no invariant geometrical meaning, but is merely a tool for the subsequent construction of
a family of finite deformations.}
\begin{equation}
z_{\ns}^{\;A}\,=\,\xi\,\@U_{\ns}^{A}\/(t)\@+\@\tfrac12\,\xi^2\,\nu_i\/\bigg(G^{AB}\,\de\@\psi^{\?i}/de{z^B}\bigg)_{\!\h\g^{\ns}}:=\,z_{\ns}^{\;A}\/(\xi,\nu,t)\@.
\end{equation}
\end{subequations}
\end{itemize}

A straightforward argument shows that for any bounded open subset $\@\Delta\subset V_h^*\@$ there exists $\?m>0\@$ such that each image
$\@\s^{\,\ns}_{\!(\xi,\nu)}\?(\pi\/(U_s)\?)\@$ is contained in $\@U_s\@$ for all $\@\nu\in \Delta\@,\;|\xi|<m\@$.

\begin{theorem}\label{TeoA.1}
Let $\@\g\@$ be an admissible piecewise differentiable evolution. Given any $\@(U,\atilde)\in\W\@$, define the functions $\@a_s\/(\xi,\nu)\@$ and the sections
$\@\s^{\,\ns}_{\!(\xi\?,\@\nu)}\@$ as above. Then, for any open bounded subset $\?\Delta\subset V_h^*\?$\vspace{.5pt} there exist an $\?\eps>0\?$ and a family
$\@\g_{(\xi\?,\@\nu)}=
\big\{\arc{\@\g^{\ns}_{(\xi,\nu)}\@}{\,a_{s-1}\/(\xi,\nu), a_s\/(\xi,\nu)}\big\}$\vspace{1.5pt} of piecewise differentiable admissible sections defined for
$\?|\?\xi\?|<\eps\@,\,\nu\in \Delta\?$ and satisfying the properties \ALPH
\begin{enumerate}
\item
$\@\g_{(0\?,\@\nu)}\/(t)\@=\@\g\/(t) \quad \forall\,\nu\@$;\vspace{3pt}
\item
$\@\g_{(\xi\?,\@\nu)}\/(t_0)\@=\@\g(t_0)\quad \forall\,\xi,\nu\@$;\vspace{4pt}
\item
$\@\g^{\ns}_{(\xi\?,\@\nu)}\?(a_s\/(\xi,\nu))= \g^{\nS}_{(\xi\?,\@\nu)}\?(a_s\/(\xi,\nu))\quad\forall\,\xi,\nu\,,\;\forall\,s=1,\ldots,N-1\@$\vspace{3pt}
\item
each arc $\,\g^{\ns}_{(\xi\?,\@\nu)}\@$ is contained in the control $\,\s^{\,\ns}_{\!(\xi\?,\@\nu)}\@$\vspace{1pt}, namely it satisfies the condition
$\@\h\g^{\ns}_{(\xi\?,\@\nu)} =\s^{\,\ns}_{\!(\xi\?,\@\nu)}\cdot\g^{\,\ns}_{(\xi\?,\@\nu)}\@$.
\end{enumerate}
\end{theorem}
\begin{proof}
As pointed out in Sec. \!2.1, each control (\ref{A.8}b) determines a \emph{velocity field\/} over $\@\pi\/(U_s)\@$, depending parametrically on $\@\xi\@$ and
$\@\nu\@$. The latter may be viewed as a vector field $\@Z_{\ns}\@=\@\de/de t\@ +\@Z_{\ns}^{\;i}\@\de/de{q^i}\@$\vspace{1pt} in the product manifold
$\@(-m,m)\times\Delta\times\pi\?(U_s)\@$, with components $\@Z_{\ns}^{\;i}\/(\xi,\nu,t,q_{\ns}^{\;i})
=\psi^i\?\bigl(\?t,\;q_{\ns}^{\;i},z_{\ns}^{\;A}\/(\xi,\nu,t)\@\bigr)\@$.

Let us denote by $\@\z^{\ns}_{\,t}\@$  the local $1$--parameter group of diffeomorphisms generated by $\@Z_{\ns}\@$, and by $\?x_s\@$ the corner
$\@\g\/(a_s)\@$.

Then, on account of Eq.\;(\ref{A.1}\@c), for any $\@\nu^*\in \Delta\@$ the orbit of $\@\z^{\ns}_{\,t}\@$ through the point $\@(0\?,\@\nu^*,x_{s-1})\@$
coincides with the coordinate line $\@q^i=0\@,\,\xi=0\@$, $\@\nu=\nu^*\!$, and is therefore defined for all $\?t\?$ in an open interval $\@(b_{s-1},b_s)\supset
[a_{s-1},a_s]\@$.\vspace{1pt}

Taking the compactness of $\@\bar \Delta\@$ into account, a standard result in classical Analysis \cite{Warner,Hurewicz} ensures the existence of a positive
$\@\eps\@$ and of open neighborhoods $\@W_{s-1}\ni x_{s-1}\@$, $\@s=1,\ldots,N\@$, such that each map $\@\z^{\ns}_{\,t}\@$ is well defined on $\@(-\eps,\eps)\times
\Delta\times W_{s-1}\@$\vspace{1pt} for all $t$ in the closed interval $\?\big[\?a_{s-1}\/(\xi,\nu),\?a_s\/(\xi,\nu)\@\big]\subset(b_{s-1},b_s)\@$.\vspace{1pt}

From this, denoting by $\@\Sigma_s\@$ the slice $\@t=a_s\/(\xi,\nu)\@$ in $\@(-m,m)\times \Delta\times\V\@$, we conclude that the group $\@\z^{\ns}_{\,t}\@$  maps the
intersection $\@W_{s-1}\cap\Sigma_{s-1}\@$ into $\@\Sigma_s\@$. Without loss of generality we may always arrange that the image of each
$\@W_{s-1}\cap\Sigma_{s-1}\@$ is contained in $\@W_s\cap\Sigma_s\@$, $\@s=1,\ldots,N\@$.

The rest is straightforward: for each $\@|\xi|<\eps\@\?,\,\nu\in \Delta\/$, consider the sequence of closed arcs
$\@\g^{\,\@\ns}_{(\xi\?,\@\nu)}:[\?a_{s-1}\/(\xi,\nu),a_s\/(\xi,\nu)\@]\to\pi\/(U_s)\@$ defined inductively by
\begin{equation*}
\begin{alignedat}{3}
& \g^{\;\,\narc{1}}_{(\xi\?,\@\nu)}\/(t)\,&&=\ \z^{\@\narc{1}}_{\@t}(\xi,\nu,\g\/(t_0))\,\qquad&& t\in[\?t_0,a_1\/(\xi,\nu)\@]\,,                   \\[3pt]
& \g^{\;\,\nS}_{(\xi\?,\@\nu)}\/(t)\,&&=\
\z^{\@\nS}_{\@t}\big(\xi,\nu,\g^{\,\@\ns}_{(\xi,\nu)}\/(a_s\/(\xi,\nu))\?\big)\,\qquad && t\in[\?a_s\/(\xi),a_{s+1}\/(\xi)\@]\,.
\end{alignedat}
\end{equation*}

The collection $\@\g\?_{(\xi\?,\@\nu)}:=\big\{\arc{\g\?^{\;\ns}_{(\xi\?,\@\nu)}}{\@a_{s-1}\/(\xi,\nu),a_s\/(\xi,\nu)}\/,\,s=1,\ldots,N\@\big\}$\vspace{1pt} is
then easily recognized to fulfil all stated requirements.
\end{proof}

In adapted coordinates, each arc $\@\g^{\ns}_{(\xi,\nu)}\@$ is represented in the form
\begin{equation}\label{A.9}
q_{\ns}^{\;i}\@=\@\vphi_{\ns}^{\;i}\/(\xi,\nu,t)\,,\qquad\;a_{s-1}\/(\xi,\nu)\le t\le a_s\/(\xi,\nu)\@,
\end{equation}
with the functions $\@\vphi_{\ns}^{\;i}\@$ satisfying the equations
\begin{subequations}\label{A.10}
\begin{equation}
\de\@\vphi_{\ns}^{\;i}/de t\,=\,\psi^i\?\biggl(\?t\,,\;\vphi_{\ns}^{\;i}\,,
\;\xi\,U_{\ns}^{A}\/(t)\@+\@\tfrac12\,\nu_k\,\xi^2\/\biggl(G^{AB}\,\de\@\psi^{\?k}/de{z^B}\biggr)_{\!\h\g^{\ns}(t)}\@\biggr)
\end{equation}
as well as the matching conditions
\begin{equation}
\vphi_{\nS}^{\;i}\big(\xi,\nu,a_s\/(\xi,\nu)\big)\@=\@q_{\nS}^{\;i}\/\big(a_s\/(\xi,\nu)\@,\@\vphi_{\ns}^{\;i}\/\big(\xi,\nu,a_s\/(\xi,\nu)\big)\big)\@,
\end{equation}
$q_{\nS}^{\;i}\@=\@q_{\nS}^{\;i}\/(t\?,\@q_{\ns}^{\;1}\@,\ldots,\@q_{\ns}^{\;n})\,$\vspace{2pt} denoting the transformation between adapted coordinates in
the intersection $\@\pi\/(U_s\cap U_{s+1}\@$).
\end{subequations}

From this, taking Eqs.\;(\ref{A.1}c) into account, it is easily seen that, for any $\@\nu^*\in\Delta\@$, the $1$--parameter family of sections
$\@\g\?_{(\xi\?,\@\nu^*)}\@$ is a deformation of $\@\g\@$, tangent to the infinitesimal deformation $\@X\@$ determined by the vector $\@(U,\atilde)\in\W\@$.

We have therefore the identifications $\@\Big(\de\vphi_{\Ns}^{\;i}\/(\xi,\nu^*,t)/de\xi\Big)_{\xi=0}\!=X_{\ns}^{\;i}\/(t)\@$\vspace{-2.5pt}, completed with the
jump relations (\ref{2.33}).\vspace{.5pt}

\smallskip
After these preliminaries, let us now focus on two facts:
\begin{itemize}
\item
on account of Theorem \ref{TeoA.1}, given any bounded open subset $\?\Delta\subset V_h^*\?$, the correspondence 
\linebreak
$\?(\xi,\nu)\to\g_{(\xi\?,\@\nu)}\/(t_1)\?$
determines a differentiable map $\@\Chi\@$ of the cartesian product $\?(-\eps,\eps)\times\Delta\?$ into the slice $\?t= t_1\?$ in $\@\V\@$, with values in a
neighborhood of the point~$\g\/(t_1)\?$;\vspace{1pt}
\item
given any differentiable curve $\@\nu=\nu\/(\xi)\@$ in $\?\Delta\@$\vspace{-1pt}, the $1$--parameter family of sections $\@\g_{\?(\xi\?,\@\nu\/(\xi))}\/(t)\@$,
$|\xi\?|<\eps\@$, $\,t\in [\?t_0,t_1\@]\@$ is a deformation of $\@\g\@$\vspace{1pt}, leaving the first endpoint $\@\g\/(t_0)\@$ fixed.
\end{itemize}

Both assertions are entirely obvious; in adapted coordinates, the map $\@\Chi\@$ is represented in the form
\begin{equation}\label{A.11}
q^{\;\,i}_{\narc{N}}\/(\?\Chi\/(\xi,\nu))\@=\@\vphi^{\;\,i}_{\narc{N}}\/(\?\xi,\nu,t_1):=\@\Chi^i\/(\?\xi,\nu)\@.
\end{equation}

Exactly as above, it may be seen that, independently of the choice of the function $\@\nu\/(\xi)\@$, the deformation
$\@\g_{\?(\xi\?,\@\nu\/(\xi))}\@$\vspace{.7pt} is always tangent to the infinitesimal deformation $\@X\@$ determined by the vector $\@(U,\atilde)\in\W\@$.

From this, taking the relations $\@\Chi^i\/(\?0,\nu)=0\vspace{1pt}\@$, $\big(\plus01\de\Chi^i/de\xi\big)_{\xi=0}=X^i\/(t_1)\@$ into account and recalling
Taylor's theorem, we conclude that, whenever the condition $\@X\/(t_1)=0\@$ holds true, i.e.\ whenever the vector $\@(U,\atilde)\@$ belongs to
$\@\ker\/(\U)\@$, the left-hand side of Eq.\;(\ref{A.11}) may be cast into the form
\begin{equation}\label{A.12}
\Chi^i\@=\@\tfrac12\,\xi^{\?2}\,\theta^{\@i}\/(\xi,\nu)\@,
\end{equation}
with $\@\theta^{\@i}\/(\xi,\nu)\@$ regular at $\@\xi=0\@$.

In this way, the original problem is reduced to establishing under what circumstances the validity of Eqs.\;(\ref{A.7}) is sufficient to ensure the existence
of an $\@\eps\@'>0\@$ and of a curve $\@\nu\/(\xi)\,$ satisfying 
\linebreak
$\@\Chi\/\big(\xi\?,\@\nu\/(\xi)\big)=\g\/(t_1)\@$
$\@\forall\,\abs\xi<\eps\@'\@$.\vspace{1.4pt}

In adapted coordinates, on account of Eq.~(\ref{A.12}), the answer relies on the solvability of the equations
\begin{equation}\label{A.13}
\theta^{\@i}\/(\xi,\nu_1,\ldots,\nu_n)\,=\,0\qquad\quad i=1,\ldots,n
\end{equation}
for the unknowns $\@\nu_i$ in a neighborhood of $\@\xi=0\@$.

To examine this aspect we notice that, in each adapted chart, Eqs.\;(\ref{A.10}a) imply the transport law
\begin{multline*}
\de/de t\bigg(\SD\text{\lower-2pt\hbox{$\@\vphi^{\;\@i}_{\ns}$}}/de /de\xi\bigg)_{\!\xi=0} \,=\,
\bigg(\SD\psi^i/de q^k/de{q^r}\bigg)_{\!\h\g^{\ns}} X^k\@X^r\@+\@2\,\bigg(\SD\psi^i/de q^k/de{z^A}\bigg)_{\!\h\g^{\ns}}X^k\@U^A\,+              \\[4pt]
+\@\bigg(\SD\psi^i/de z^A/de{z^B}\bigg)_{\!\h\g^{\ns}}\@U^A\@U^B\,+\@
\cancel{\bigg(\de\psi^i/de{q^k}\bigg)_{\!\h\g^{\ns}}\!\bigg(\SD\text{\lower-2pt\hbox{$\@\vphi^{\;\@k}_{\ns}$}}/de /de\xi\bigg)}{\plus0{10}}_{\!\xi=0}
+\,\bigg(\de\psi^i/de{z^A}\bigg)_{\!\h\g^{\ns}}\/\nu_{\?k}\/\biggl(G^{AB}\,\de\@\psi^{\?k}/de{z^B}\biggr)_{\!\h\g^{\ns}}\,,
\end{multline*}
the cancelation arising from Eq.\;(\ref{A.1}\@c).

In a similar way, from the matching conditions (\ref{A.10}b), recalling Eqs.\;(\ref{A.3}), (\ref{A.8}a) and evaluating everything at $\?\xi=0\@$, we get the
relations
\begin{multline*}
\biggl[\biggl(\SD\vphi_{\nS}^{\;i}/de /de\xi\biggr)_{\!\xi=0}\!-\biggl(\SD\vphi_{\ns}^{\;i}/de /de\xi\biggr)_{\!\xi=0}\biggr]_{x_s}\@=
\@\a_s^{\,2}\,\SD q_{\nS}^{\;i}/de /de t\@+\@2\@\a_s\,\SD q_{\nS}^{\;i}/de t\;/de{q_{\ns}^{\;k}}\@\,X_{\ns}^{\;k}\,+                      \\[3pt]
+\SD q_{\nS}^{\;i}/de{q_{\ns}^{\;h}}\,/de{q_{\ns}^{\;k}}\,X_{\ns}^{\;h}\@X_{\ns}^{\;k}\@-\@2\@\a_s\@\biggl[\d X_{\nS}^{\;i}/dt -
\d X_{\ns}^{\;i}/dt\biggr]_{x_s}\!\@+ \a_s^{\,2}\,\Big[\?\psi^i\/(\h\g)\?\Big]_{\?x_s}\@\Big[\?\psi^k\/(\h\g)\?\Big]_{\?x_s}\@\nu_k
\end{multline*}
expressing the jumps $\@\Big[\Big(\sd\vphi_{\NS}^{\;i}/de /de\xi\Big)_{\!\xi=0}\!- \Big(\sd\vphi_{\Ns}^{\;i}/de /de\xi\Big)_{\!\xi=0}\Big]_{\@x_s}\!$ at the
corners.\vspace{3pt}

Collecting all results and recalling Eqs.\;(\ref{A.11}), (\ref{A.12}) we obtain the expression
\begin{align}\label{A.14}
&\hskip-.5em\theta^{\?i}\@\big|_{\@\xi=0}\,=\,\bigg(\SD\Chi^i/de /de\xi\bigg)_{\!\xi=0}\,=  \nonumber\\[3pt]
&\hskip-.7em=b\@^i\!+\!\left(\@\sum_{s=1}^N\int_{a_{s-1}}^{a_s}\!\!\!\bigg(G^{AB}\,\de\psi^i/de{z^A}\,\de\@\psi^{\?k}/de{z^B}\bigg)_{\!\h\g^{\ns}}d\/t\@+\!
\sum_{s=1}^{N-1}\a_s^{\,2}\Big[\?\psi^i\/(\h\g)\?\Big]_{a_s}\!\@\Big[\?\psi^k\/(\h\g)\?\Big]_{a_s}\!\right)\!\nu_k
\end{align}
with $\@b\@^i\in\R\@$ independent of $\@\nu\@$. We can therefore state
\begin{proposition}\label{ProA.2}
Let $\@\g:[\?t_0,t_1\?]\to\V\@$ be a continuous, piecewise differentiable, admissible section. Then, if the matrix
\begin{equation}\label{A.15}
S^{\?ik}\,:=\,\int_{t_0}^{t_1}\,\bigg(G^{AB}\,\de\psi^i/de{z^A}\,\de\@\psi^{\?k}/de{z^B}\bigg)_{\!\h\g}\,d\/t\,+\,
\sum_{s=1}^{N-1}\a_s^{\,2}\@\Big[\?\psi^i\/(\h\g)\?\Big]_{a_s}\Big[\?\psi^k\/(\h\g)\?\Big]_{a_s}
\end{equation}
is non--singular, every infinitesimal deformation of $\@\g\@$ vanishing at the endpoints is tangent to a finite deformation with fixed endpoints.
\end{proposition}
\begin{proof}
The conclusion follows at once from the fact that, on account of Eq.\;(\ref{A.14}), the non--singularity of the matrix (\ref{A.15}) ensures the solvability of
Eq.\;(\ref{A.13}) in a neighborhood of $\@\xi=0\,$, $\@\forall\,\nu\@$.
\end{proof}
Proposition \ref{ProA.2} may be rephrased in the language of Sec. \!\ref{Sec2.7} observing that, whenever the section $\@\g\@$ is abnormal, the matrix (\ref{A.15})
is necessarily \emph{singular\/}.

In the stated circumstance, in fact, Proposition \ref{Pro2.4} and Eq.\;(\ref{A.6}) imply the existence of at least one non--zero virtual $1$--form
$\@\rho{\?_i}\,\wg i\@$ with \emph{constant\/} components $\rho_{\?i}\?$ obeying the relations
\begin{equation}\label{A.16}
\rho_{\?i}\,\bigg(\de\psi^i/de{z^A}\BIGG)_{\h\g\/(t)}\,=\,0\,,\qquad
\rho_{\?i}\@\big[\@\psi^i\/(\h\g)\@\big]_{\?a_s}\,=\,0
\end{equation}
and therefore automatically satisfying $\@\rho_{\?i}\,S^{\?ij}=0\@$.

More specifically, denoting by $\/p\@$ the \emph{abnormality index\/} of $\@\g\@$, we have the following
\begin{theorem}\label{TeoA.2}
The matrix (\ref{A.15}) has rank $\/n-p\@$.
\end{theorem}
\begin{proof}
By definition, the index $p\@$ coincides with the dimension of the annihilator $\@\big(\@\U\?(\W)\@\big){}^0\subset V_h^*\@$\vspace{1pt} which, in turn, is
identical to the dimension of the space of constant solutions of Eqs.\;(\ref{A.16}).

On the other hand, by Eq.\;(\ref{A.15}), the matrix $\@S^{\?ij}\@$ is positive semidefinite. Its kernel is therefore identical to the totality of zeroes of the
quadratic form $\@S^{ij}\rho_{\?i}\@\rho_j\@$, i.e.\ to the totality of $n$--tuples $\@(\rho_{\?1}\@,\ldots,\rho_{\?n}) \in\R^n\@$ satisfying the relation
\begin{multline*}
0=\left(\int_{t_0}^{t_1}\,\bigg(G^{AB}\,\de\psi^i/de{z^A}\,\de\@\psi^{\?k}/de{z^B}\bigg)_{\!\h\g}\,d\/t\,+\,
\sum_{s=1}^{N-1}\a_s^{\,2}\@\Big[\?\psi^i\/(\h\g)\?\Big]_{a_s}\Big[\?\psi^k\/(\h\g)\?\Big]_{a_s}\@\right)\rho_{\?i}\@\rho_k\@=                 \\[3pt]
=\@\int_{t_0}^{t_1}\!\bigg(G^{AB}\,\rho_{\?i}\,\de\psi^i/de{z^A}\,\rho_{\?k}\de\@\psi^{\?k}/de{z^B}\bigg)_{\!\h\g}\,d\/t\,+\,
\sum_{s=1}^{N-1}\a_s^{\,2}\@\bigg(\?\rho_{\?i}\Big[\?\psi^i\/(\h\g)\?\Big]_{a_s}\bigg)^2\@.
\end{multline*}

Due to the positive definiteness of $\@G^{AB}\/(t)\@$, the last condition is equivalent to Eqs.\;(\ref{A.16}). This proves
$\@\dim\big(\?\ker(\?S^{ij}\?\big)\big)=p\@$ $\;\Rightarrow\;\;\rank\big(\?S^{ij}\?\big)=n-p\@$.
\end{proof}

Proposition \ref{ProA.2} and Theorem \ref{TeoA.2} show that the normal evolutions form a subset of the ordinary ones, thus establishing
Proposition~\ref{Pro2.5} of \@Sec. \ref{Sec2.7}.

Along the same lines, a useful generalization is provided by the following
\begin{theorem}\label{TeoA.3}
Let $p\;\,(\?\ge 0\@)\@$ denote the abnormality index of $\@\g\@$. Then a sufficient condition for every infinitesimal deformation vanishing at the endpoints
of $\?\g\?$ to be tangent to a finite deformation with fixed endpoints is the existence of an
$\/(n\!-\!p)$--dimensional submanifold $\@S\subset\V\@$ such that every deformation $\@\g_{\xi}\@$ leaving $\g\/(t_0)\@$ fixed satisfies $\@\g_{\xi}\/(t_1)\in
S\@$ for $\@\xi\@$ sufficiently small.
\end{theorem}
\begin{proof}
Using the freedom expressed by Eq.~(\ref{A.2}), we choose the adapted coordinates in such a way that, in a neighborhood of $\@\g\/(t_1)\@$, the submanifold
$\@S\@$ has local equation $\@q_\narc{N}^{\@p+1}=\cdots=q_\narc{N}^n=0\@$
\footnote%
{\,Namely, we choose the last chart $\@(U_N,k_N)\@$ consistently with the stated condition, ant then proceed backwards, requiring the validity of
Eqs.~(\ref{A.4}) at the corners.}.

Given an infinitesimal deformation $\@X\@$ generated by an element $\@(U,\atilde)\in\W\@$, we then proceed exactly as above, ending up with a finite
deformation $\@\g\?_{(\xi\?,\@\nu)}\@$ tangent to $\@X\@$, described in coordinates as 
\linebreak
$\@q^{\;\,i}_{\ns}=\vphi^{\;\,i}_{\ns}\/(\?\xi,\nu,t)\@$,
$\@a_{s-1}\/(\xi,\nu)\le t\le a_s\/(\xi,\nu)\@$, with the functions $\@\vphi^{\;\,i}_{\ns}\/(\?\xi,\nu,t)\@$ satisfying all conditions of Theorem
\ref{TeoA.1}.\vspace{1.2pt}

Once again, we focus on the ``end--point map'' $\?\Chi\/(\xi,\nu)=\g_{(\xi\?,\@\nu)}\/(t_1)\?$ and on the fact that, under the assumption $\@X\/(t_1)=0\@$, the
functions $\@\Chi\@^i\/(\?\xi,\nu)=\vphi^{\;\,i}_{(N)}\/(\?\xi,\nu,t_1)\@$ factorize into the product
\begin{equation*}
\Chi^i\@=\@\tfrac12\,\@\xi^{\?2}\,\theta^{\@i}\/(\xi,\nu)\@,
\end{equation*}
with $\@\theta^{\?i}\@\big|_{\@\xi=0}\@$ expressed as a linear polynomial
\begin{equation*}
\hskip-.5em\theta^{\?i}\@\big|_{\@\xi=0}\,=\,\bigg(\SD\Chi^i/de /de\xi\bigg)_{\!\xi=0}\,=b\@^i+ S^{\@ij}\nu_j\@.
\end{equation*}

Actually, in the present case, the condition $\@\Chi^{\,\a}\/(\?\xi,\nu)=0\@$, $\@\a=p+1,\ldots,n\@$ implies $\@\theta^{\@\a}\/(\?\xi,\nu)=0\@$ whence, in
particular, $\@b^{\@\a}=S^{\@\a j}=0\@$.

At the same time, Theorem \ref{TeoA.2} ensures the equality between $\@\rank S^{\?ij}\@$ and the abnormality index $\@p\@$ of $\@\g\@$.

In the case in study, the matrix $\@S^{\?ij}\@$ is therefore necessarily of the form
\begin{equation*}
S^{\?ij}=\,
\begin{pmatrix}
  S^{\?AB} & 0 \\
  0\quad & 0
\end{pmatrix}\,,
\qquad A,B=1,\ldots,p\,,
\end{equation*}
with $\@\det S^{\?AB}\ne0\@$.

The rest is now straightforward: in order to establish the existence of a finite deformation with fixed endpoints tangent to $\@X\@$ we have to verify that the
equations $\@\theta^{\@i}\/(\xi,\nu)=0\@$ admit at least one solution $\@\nu=\nu\/(\xi)\@$ in a neighborhood of $\@\xi=0\@$. And indeed, no matter how we
choose the functions $\@\nu_\a\/(\?\xi)\@$, $\@\a=p+1\,\ldots,n\@$, the $\@n-p\@$ equations $\@\theta^{\,\a}\/(\?\xi,\nu)=0\@$ are identically satisfied, while
the remaining ones form a system of $\@p\@$ equations for the unknowns $\@\nu_1,\ldots,\nu_p\@$, whose solvability is ensured by the non singularity of the
Jacobian $\@\de\theta^{\?A}/de{\nu_B}\big|_{\@\xi=0}=S^{\?AB}$.
\end{proof}

\smallskip
\section{Comments on normality}\label{SecB}
The following arguments help clarifying some aspects of the concept of normality discussed in 
\linebreak
Sec. \!2.7.

Let $\@\g=\big\{\arc{\?\g^{\ns}}{\?a_{s-1},a_s}\/\big\}\@$ be a piecewise differentiable evolution. According to Proposition \ref{Pro2.4}, if at least one arc
$\@\g^{\ns}\@$ is normal, $\@\g\@$ is necessarily normal.

More generally, an evolution may happen to be normal even when \emph{all\/} its arcs $\@\g^{\ns}\@$ are abnormal. Examples in this sense are:
\Tondo
$\@\V=\R\times E_2\@$, referred to coordinates $\@t,x,y\@$. Constraint: $\@\dot x^2+ \dot y^2=v^2$. Imbedding $\@\A\to\j\V\@$ expressed in coordinates as $\@\dot x=
v\@\cos z\@,\,\dot y=v\@\sin z\@$.
Piecewise differentiable evolution $\@\g\@$ consisting of two arcs:
\begin{equation*}
\begin{alignedat}{3}
   &\g^{\narc{1}}:\quad x=0\@, &&y=v\?t\qquad &&t_0\le t\le 0                                                                        \\
   & \g^{\narc{2}}:\quad x=v\?t\@,\quad&&y=0\qquad\; &&\;0\le t\le t_1
\end{alignedat}
\end{equation*}
Eq.~(\ref{2.42}a) admits $h$--transported solutions $\?\rho^{\narc{1}}\!=\a\@\delta y_{\?|\g}\@$, $\rho^{\narc{2}}\!=\b\@\delta x_{\?|\g}$ ($\a,\b\in\R$)
respectively along $\?\g^{\narc{1}}\?$ and $\?\g^{\narc{2}}$: both arcs are therefore abnormal. Nevertheless $\?\g\?$ is normal, since no pair
$\@\rho^{\narc{1}}\!,\rho^{\narc{2}}\@$ matches into a continuous non--zero virtual $1$--form along $\g$.\vspace{4pt}

\Tondo
$\@\V=\R\times E_2\@$, referred to coordinates $\@t,x,y\@$\vspace{2pt}. Constraint: $\@v^3\@\dot x=(\dot y^2-a^2\@t^2)^2\@$. Imbedding 
\linebreak
$\@\A\to\j\V\@$ expressed in
coordinates as $\@\dot x=v^{-3}\@(z^2-a^2\@t^2)^2\@$, $\dot y=z\@$. Piecewise differentiable evolution $\@\g\@$ consisting of two arcs:
\begin{equation*}
\begin{alignedat}{3}
&\g^{\narc{1}}:\quad x=0\@, &&y=\frac12\,a\@(t^2-t^{*\@2})\qquad &&t_0\le t\le t^*                                                     \\
& \g^{\narc{2}}:\quad x=\frac{a^4}{5\@v^3}\,\@(t^5-t^{*\@5})\@,  \qquad&&y=0 &&t^*\le t\le t_1
\end{alignedat}
\end{equation*}
$(t^*\ne 0\@$). Eq.~(\ref{2.42}a) admits $h$--transported solutions of the form $\@\rho=\a\@\delta x_{\?|\g}\@$ along the whole of $\@\g\@$. Both arcs
$\@\g^{\narc{1}}\!,\,\g^{\narc{2}}\@$ are therefore abnormal. Nevertheless, $\?\g\?$ is normal, since no solution satisfies condition (\ref{2.42}b).\vspace{4pt}

\Tondo
By definition, local normality implies normality. The converse is generally false, as shown by the previous examples. A further example, not involving the
presence of corners, is the following: let the imbedding $\@\A\rarw{i}\j\V\@$ be locally described by the equations
\begin{equation*}
\left\{
\begin{alignedat}{2}
  & \dot q^A&&=z^A\qquad\qquad A=1,\ldots, n-1  \\
  & \dot q^n&&=f\/(t)\,z^1
\end{alignedat}
\right.
\end{equation*}
with $\,f(t)=\exp\?(-\@\nicefrac1{\plus60 t^2})\,$ for $\@t<0\,$ \vspace{.8pt} and $\,f(t)=0\,$ for $\@t\geq 0\@$.

Along any admissible section $\@\g\colon[\@t_0,t_1]\to\V\@$, the condition of $\@h$--transport and Eqs.~\!(\ref{2.42}\@a) are summarized into the relations
\begin{equation}\label{b.1}
\d\l_i/d t\,=\,0\,,\qquad\;\l_1\@+\@\l_n\@f\/(t)\,=\,0\,,\qquad\l_2\,=\,\cdots\,=\,\l_{n-1}\,=\,0\,.
\end{equation}

In particular, if $\@t_0<0<t_1\@$ we conclude that:\vspace{-2pt}
\begin{list}{$\diamond$}{}
\item
$\@\g\@$ is normal, since Eqs.~\!(\ref{b.1}) do not admit any non-zero solution in $\@[\@t_0,t_1\?]\@$;\vspace{2pt}
\item
$\@\g\@$ is not locally normal, since Eqs.~\!(\ref{b.1}) admit the solution $\@\l_A=0\@$, $\@\l_n=\const\@$ in any subinterval $[\@a,b\@]\subseteq[\@0,t_1]\@$.
\end{list}

\end{document}